\documentclass[10pt,pra,superscriptaddress,amsmath,amssymb,twocolumn]{revtex4}

\usepackage{amsthm}
\usepackage{graphicx}
\usepackage[all]{xy}

\newcommand{\chan}[1]{\mathcal #1}
\newcommand{\cchan}[1]{\widehat {\mathcal #1}}
\newcommand{\gcchan}[1]{\widetilde {\mathcal #1}}

\newcommand{\comp}{}

\newcommand{\bra}[1]{\langle #1 |}
\newcommand{\ket}[1]{| #1 \rangle}

\newcommand{\ketbra}[2]{\ket{#1}\bra{#2}}
\newcommand{\proj}[1]{\ket{#1}\bra{#1}}
\newcommand{\tr}{{\rm Tr}}
\newcommand{\one}{{\bf 1}}
\newcommand{\re}{{\rm Re}\,}

\newcommand{\lnt}{\succeq}   

\newcommand{\good}[1]{#1_g}

\newcommand{\Hil}{\mathcal H}
\newcommand{\chanh}[1]{\mathcal #1^\dagger}
\newcommand{\cchanh}[1]{{\cchan #1}^{{\dagger}}}

\newcommand{\id}{{\rm id}}

\theoremstyle{plain}
\newtheorem{proposition}{Proposition}
\newtheorem{theorem}[proposition]{Theorem}
\newtheorem{lemma}[proposition]{Lemma}
\newtheorem{corollary}[proposition]{Corollary}

\theoremstyle{definition}
\newtheorem{definition}{Definition}

\begin{document}

\title{Approximate simulation of quantum channels}
\author{C\'edric B\'eny}
\affiliation{Centre for Quantum Technologies, National University of Singapore, 3 Science Drive 2, Singapore 117543}
\affiliation{Institut f\"ur Theoretische Physik, Leibniz Universit\"at Hannover Appelstra\ss e 2, 30167 Hannover, Germany}
\author{Ognyan Oreshkov}
\affiliation{QuIC, Ecole Polytechnique, CP 165, Universit\'{e} Libre de Bruxelles, 1050 Brussels, Belgium}
\date{\today}

\begin{abstract}
In Ref.~\cite{beny10}, we proved a duality between two optimizations problems. The primary one is, given two quantum channels $\chan M$ and $\chan N$, to find a quantum channel $\chan R$ such that $\chan R \circ \chan N$ is optimally close to $\chan M$ as measured by the worst-case entanglement fidelity. The dual problem involves the information obtained by the environment through the so-called complementary channels $\cchan M$ and $\cchan N$, and consists in finding a quantum channel $\chan R'$ such that $\chan R' \circ \cchan M$ is optimally close to $\cchan N$. It turns out to be easier to find an approximate solution to the dual problem in certain important situations, notably when $\chan M$ is the identity channel---the problem of quantum error correction---yielding a good near-optimal worst-case entanglement fidelity as well as the corresponding near-optimal correcting channel. Here we provide more detailed proofs of these results. In addition, we generalize the main theorem to the case where there are certain constraints on the implementation of $\chan R$, namely on the number of Kraus operators. We also offer a simple algebraic form for the near-optimal correction channel in the case $\chan M = \id$. For approximate error correction, we show that any $\varepsilon$-correctable channel is, up to appending an ancilla, $\varepsilon$-close to an exactly correctable one. We also demonstrate an application of our theorem to the problem of minimax state discrimination.
\end{abstract}

\maketitle



Shannon theory and error correction, be their classical or quantum version, are based on the problem of transmitting information through a given noisy channel $\chan N$ by choosing an encoding channel $\chan E$ and a decoding channel $\chan R$ such that when composed, they simulate a noiseless channel $\id$: $\chan R \chan N \chan E \approx \id$. For instance, the capacity of a channel $\chan N$ is the largest ratio $n/m$ such that $\chan N^{\otimes m}$ can perfectly simulate $\id^{\otimes n}$ in the limit $n \rightarrow \infty$~\cite{kretschmann04}.





Classically, the problem of simulating a noiseless channel has been found to provide the benchmark for most of a channel's information carrying capabilities.
In quantum information theory, however, the situation is more complex. For instance, the capacity of a quantum channel defined in this way---the quantum capacity---does not suffice to determine the capacity of channels used in conjunction with each other~\cite{smith08}.

More generally, one can consider the simulation of an arbitrary channel $\chan M$: $\chan R \chan N \chan E \approx \chan M$. For example, if $\chan N$ is quantum, choosing the target $\chan M = \id^{\otimes n}$ yields the quantum capacity, while using copies of a fully decoherent (i.e., classical) channel
\begin{equation}
\chan M(\rho) = \sum_i \bra i \rho \ket i \proj i
\end{equation}
yields the classical capacity.
A slightly more general case is that where $\chan M$ is a noiseless channel on any C$^*$-algebra, which yields subsystem quantum error correction (QEC) and hybrid quantum-classical error correction~\cite{beny07x1}, or hybrid capacities in the asymptotic case~\cite{devetak05}.


A fundamental result in the case of standard QEC, namely, simulation of the identity quantum channel, is the Knill-Laflamme conditions~\cite{knill97}, which, given an encoding $\chan E$, provide a criterion for the existence of a corresponding decoding channel $\chan R$. Specifically, it says that $\chan R$ exists, i.e., the channel and code are correctable, if and only if the environment gains no information about the encoded state. This condition, in its approximate form~\cite{kretschmann08}
is also the main starting point for the ``decoupling'' approach to channel capacities~\cite{hayden05}, namely, the corresponding result applied to states via the Choi-Jamio{\l}kowski isomorphism.

Here we detail and extend work presented in Ref.~\cite{beny10} which generalizes these results to a generic target $\chan M$. In Section~\ref{main-thm-section}, we prove two slightly more general versions of the main theorem of Ref.~\cite{beny10}, which says that
the optimal distance (optimized over the decoding operations $\chan R$) between $\chan R \chan N$ and $\chan M$ is equal to the optimal distance between $\cchan N$ and $\chan R\cchan M$, where the hat denotes respective {\em complementary channels}.
This yields an efficient way of estimating the optimal distance for a large class of target channels $\chan M$ (Section~\ref{special-case-bounds}). We also show how to explicitly construct a channel $\chan R$ achieving the estimated distance (Section~\ref{special-case-channel}).
In Section~\ref{nature}, we show that any $\varepsilon$-correctable channel is, up to appending an ancilla, $\varepsilon$-close to an exactly correctable one. We also give an application of our main theorem to the problem of minimax state discrimination in Section~\ref{discrimination}.


We note that our approach also yields new results in the important special case $\chan M = \id$. In order to measure the quality of a simulation, we use a fidelity-based distance as in Refs.~\cite{barnum00,tyson09x1}. In contrast to these works, however, our approach yields an approximate reversal channel for the {\em worst-case} entanglement fidelity which is a state-independent measure. See Section~\ref{comparison} for a comparison with these works in the case $\chan M = \id$.

Other works focused on the worst-case trace distance~\cite{kretschmann08, beny09x1}. The advantage of the fidelity over these is that the optimal fidelity in the dual problem is precisely equal to the optimal fidelity in the original problem. Although the dual problem may not be solved easily, it can be precisely estimated. This yields better bounds on the original optimization which are
useful not only for analyzing asymptotic scenarios,
but may be advantageous also in one-shot scenarios, e.g., to estimate the error in an approximate quantum error correction scheme~\cite{leung97,schumacher01,crepeau05,buscemi08,ng09,beny11}.


In addition, since our results concern only the problem of finding $\chan R$ such that $\chan R \chan N' \approx M$, they are not restricted to the case where $\chan N' = \chan N \chan E$ but apply also to the case where resources are shared between the sender and the recipient. For instance, for an entanglement-assisted scheme \cite{bennett99, brun06} one would consider a channel of the form
\begin{equation}
\chan N'(\rho) = \frac 1 d \sum_i (\chan N \chan E) (\rho \otimes \ketbra i j) \otimes \ketbra i j.
\end{equation}




\section{Preliminaries}

In this paper, we restrict our considerations to finite-dimensional systems. We begin by introducing the main concepts that we will be using to express our results.

\subsection{Complementary channels}

A channel ${\chan N}$ is a completely positive trace-preserving map. It can always be written as
\begin{equation}
{\chan N}(\rho) = \sum_i E_i \rho E_i^\dagger,
\end{equation}
where the Kraus operators $E_i$ satisfy $\sum_i E_i^\dagger E_i = \one$. Conversely, any function of this form is completely positive and trace-preserving.
The dual ${\chan N}^\dagger$ is defined by the relation
\begin{equation}
\tr({\chan N}(\rho) A ) = \tr(\rho \, {\chan N}^\dagger (A))
\end{equation}
for any state $\rho$ and any operator $A$. This implies that
\begin{equation}
{\chanh N}(A) = \sum_i E_i^\dagger \rho E_i.
\end{equation}
Physically, ${\chan N}$ is interpreted as evolving states, while ${\chan N}^\dagger$ evolves observables. Hence, ${\chan N}^\dagger$ represents the Heisenberg picture for the evolution defined by the channel. To avoid confusion, we only call ${\chan N}$ a {\em channel}, while ${\chan N}^\dagger$ is its {\em dual}.

Note that the dimension of the input and output quantum systems may differ. In this case the operators $E_i$ are not square matrices. For instance, there is only one channel whose output Hilbert space has dimension one: the trace. Channels with a one-dimensional input are in one-to-one correspondence with quantum states. Since the input is just a complex number $z$, the output can only be of the form $z \rho$ for some fixed quantum state $\rho$.

We let the reader check that the dual of the trace channel takes as an input a complex number $\alpha$ (an operator on the one-dimensional Hilbert space), and outputs the identity operator times $\alpha$:
\begin{equation}
\tr^\dagger(\alpha) = \alpha \one.
\end{equation}
Note that the partial trace can be written as $\id \otimes \tr$, where $\id$ is the identity channel on the subsystem that is not traced over.
Hence, for instance,
\begin{equation}
(\id \otimes \tr)^\dagger(A) = A \otimes \one.
\end{equation}

Stinespring's dilation theorem~\cite{stinespring55} guarantees that for any channel ${\chan N}$ we can find a (nonunique) {\em isometry} $V$ which maps the input space of $\chan N$ into its output space extended by an extra system $E$, the ``environment'', such that
\begin{equation}
\label{dilation}
{\chanh N}(A) = V^\dagger (A \otimes \one_E) V,
\end{equation}
or equivalently,
\begin{equation}
{\chan N}(\rho) = (\id \otimes \tr_E) (V \rho V^\dagger),
\end{equation}
where $\id \otimes \tr_E$ denotes the partial trace over $E$.
We say that $V$ defines a \textit{dilation} of ${\chan N}$.

An isometry $V$ is any operator satisfying the property $V^\dagger V = \one$. It describes how the input Hilbert space is isometrically embedded in the target space. It can be seen as a unitary operator but with the input restricted to a subspace. For instance, if the target space dimension is divisible by the input space dimension, it amounts to adding an auxiliary system with a fixed pure initial state $\ket{\phi_0}$ and letting it interact unitarily with the system, i.e.,
\begin{equation}
V \ket{\psi} := U (\ket{\psi} \otimes \ket{\phi_0})
\end{equation}
for some unitary operator $U$.

It is easy to see that this isometry $V$ is not unique. Indeed we can always use another isometry $W$ from $E$ to a larger environment $E'$ to obtain also
\begin{equation}
{\chan N}(\rho) = (\id \otimes \tr_{E'}) (V' \rho (V')^\dagger),
\end{equation}
where $V'$ is the new isometry
\begin{equation}
V' = (\one \otimes W) V.
\end{equation}

\begin{definition}
Let $V$ define a dilation of the channel $\chan N$ as above. We say that the channel $\cchan N$, defined by
\begin{equation}
\cchanh N(B) =  V^\dagger (\one \otimes B) V
\end{equation}
for all $B$, is {\em complementary} to $\chan N$.
\end{definition}
Equivalently, we have
\begin{equation}
{\cchan N}(\rho) = (\tr \otimes \id) (V \rho V^\dagger).
\end{equation}
The channel $\cchan N$ maps the initial state of the system to the final state of the environment.

It is clear from the definition that $\chan N$ is also complementary to $\cchan N$.

Since the isometry $V$ associated with $\chan N$ is not unique as observed above, there are correspondingly many channels complementary to $\chan N$. For instance, if $\cchan N$ maps the input system to the environment $E$, and $W$ is an isometry from $E$ to $E'$, then the channel $\chan N'$ defined by
\begin{equation}
(\cchan N')^\dagger(B) = (V')^\dagger (\one \otimes B) V',
\end{equation}
where
\(
V' = (\one \otimes W) V,
\)
is also complementary to $\chan N$.



The connection between the isometry $V$ and the channels' Kraus operators can be found by introducing an orthonormal basis $\ket{i}$ of the environment as follows:
\begin{equation}
\begin{split}
{\chanh N}(A) &= V^\dagger (A \otimes \one) V = \sum_i V^\dagger (A \otimes \proj{i}) V\\
&= \sum_i V^\dagger (\one \otimes \ket{i}) A  (\one \otimes \bra{i}) V.
\end{split}
\end{equation}
Hence, we can use
\begin{equation}
E_i = (\one \otimes \bra{i}) V,
\end{equation}
which is defined by
\begin{equation}
\bra{\psi} E_i = (\bra{\psi} \otimes \bra{i}) V.
\end{equation}
This implies that the complementary channel associated with the isometry $V$ can be written in a dual form as
\begin{equation}
\cchanh N(B) = \sum_{ij} \bra{i} B \ket{j} E_i^\dagger E_j.
\end{equation}


\subsection{Minimal dimension of the environment}

\begin{definition}
Given a channel $\chan N$, we write $|\chan N|$ for the minimal number of Kraus operators with which it can be represented. This corresponds to the minimal dimension of the environment for which it has an isometric implementation (Eq.~\eqref{dilation}).
\end{definition}

If the channel $\chan N$ maps states over the Hilbert space $\Hil_A$ to states over $\Hil_B$, then \cite{choi75}
\begin{equation}
1 \le |\chan N| \le \dim(\Hil_A)\dim(\Hil_B).
\end{equation}
In addition, $|\chan N| = 1$ if and only if
\begin{equation}
\chan N(\rho) = V \rho V^\dagger
\end{equation}
for some isometry $V$.

\subsection{Postprocessing order and equivalence relation}

\begin{definition}
Given two channels $\chan N$ and $\chan M$, we write
\begin{equation}
\chan N \lnt \chan M
\end{equation}
if there exists a channel $\chan R$ such that $\chan R \chan N = \chan M$ (i.e., $\chan R(\chan N(\rho)) = \chan M(\rho)$ for all states $\rho$).
We also write
\begin{equation}
\chan N \sim \chan M
\end{equation}
when we have both $\chan N \lnt \chan M$ and $\chan M \lnt \chan N$.
\end{definition}

Note that this relation $\lnt$ can equivalently be defined dropping the requirement that $\chan R$ be trace-preserving. Indeed, $\chan R$ can always be completed to the trace-preserving channel
\begin{equation}
\chan R'(\rho) = \chan R(\rho) + \tr[(\one-\chanh R(\one))\rho] \sigma,
\end{equation}
where $\sigma$ is an arbitrary state, such that $\chan R'\chan N = \chan M$.

It is easy to see that the relation $\lnt$ is a preorder on all channels, and that $\sim$ is an equivalence relation.

The relation $\chan N \sim \chan M$ can be seen as a precise way of saying that the two channels carry the same information about the initial system, independently of further processing. For this reason, we will mostly focus on the equivalence classes rather than on individual channels.

Note that $\chan N \sim \chan M$ does not imply that $\chan N$ and $\chan M$ are related by a unitary map. For instance, consider the family of channels $\chan S_{\tau}(\rho) := \tr(\rho) \tau$. These channels are all equivalent since $\chan S_\tau = \chan S_\tau \chan S_{\tau'}$. But they can be related by a unitary map if and only if $\tau$ and $\tau'$ have the same spectrum. 

Let us state a few elementary facts in order to build up some intuition about this relation. All channels are bounded from above by the unitary channels (in particular), which are all equivalent to the identity channel $\id$, and from below by the channels $S_\tau$ mentioned above, which are all equivalent to the trace $\tr$ (which is the only channel with a target space of dimension $1$):
\begin{equation}
\id \lnt \chan N \lnt \tr.
\end{equation}

Note that for any channels $\chan N$, $\chan M$,  $\chan N'$, and $\chan M'$,
\begin{equation}
\chan N \otimes \chan N' \lnt \chan M \otimes \chan M' \Longleftrightarrow \chan N \lnt \chan M \text{ and } \chan N' \lnt \chan M'.
\end{equation}
In particular,
\begin{equation}
\chan N \lnt \chan M \Longleftrightarrow \chan N \otimes \id \lnt \chan M \otimes \id.
\end{equation}
Since channels are trace-preserving, we always have
\begin{equation}
\tr \comp \chan N = \tr.
\end{equation}

Let us show that all the channels complementary to a given one belong to the same equivalence class.
\begin{lemma}
If $\cchan N$ and $\cchan N'$ are both complementary to $\chan N$, then $\cchan N \sim \cchan N'$.
\end{lemma}
\proof
Suppose without loss of generality that the dimension of the output of $\cchan N'$ is larger or equal to that of $\cchan N$. Then there is an isometry $W$ such that
\begin{equation}
\cchan N'(\rho) = W \cchan N(\rho) W^\dagger
\end{equation}
for all $\rho$. Hence, clearly, $\cchan N' \lnt \cchan N$. In order to show that also $\cchan N \lnt \cchan N'$, we need to build a channel $\chan R$ such that $\cchan N = \chan R \cchan N'$. Since $V^\dagger V = \one$, we would like to use the completely positive map $\rho \mapsto V^\dagger \rho V$. Unfortunately, this map is not trace-preserving. Instead, letting $P := VV^\dagger$, we use
\begin{equation}
\chan R(\rho) = V^\dagger \rho V + \tr((\one - P) \rho).
\end{equation}
It is easy to check that this map is trace-preserving. Furthermore, since $(\one - P) V = 0$, the extra term does not affect the ability of $\chan R$ to invert $\cchan N'$.
\qed

Note that the equivalence class associated with $\cchan N$ is in general larger than the set of channels complementary to $\chan N$.  We will regard the channels belonging to this class as \textit{generalized} complementary channels and will typically denote them by a tilde, i.e., as $\gcchan{N}$.
\begin{definition}
The channel $\gcchan{N}$ is a generalized complementary channel of $\chan N$ if $\gcchan{N}\sim {\cchan N}$, where ${\cchan N}$ is complementary to $\chan N$.
\end{definition}

The generalized complementary channels satisfy the following important property:

\begin{theorem}
If $\gcchan{N}$ and $\gcchan{M}$ are generalized complementary channels of $\chan N$ and $\chan M$, respectively, then
\label{fund}
\begin{equation}
\chan N \lnt \chan M \Longleftrightarrow \gcchan{M} \lnt \gcchan{N}.
\end{equation}
\end{theorem}
\proof
It is sufficient to show that $\chan N \lnt \chan M \Longrightarrow \gcchan{M} \lnt \gcchan{N}$.
For any channel $\chan N$, we will write
\begin{equation}
\chan V_{\chan N}(\rho) = V_{\chan N} \rho V_{\chan N}^\dagger
\end{equation}
for an isometric map characterizing a dilation of $\chan N$, i.e., such that
\(
\chan N = (\id \otimes \tr) \chan V_{\chan N}.
\)
Note that if $\gcchan{R}$ is a generalized complementary channel of $\chan R$,
then
\(
(\gcchan{R} \otimes \id) \comp \chan V_{\chan N}
\)
is a generalized complementary channel of $\chan R \comp \chan N$.

If $\chan N \lnt \chan M$, then there is a channel $\chan R$ such that $\chan R \chan N = \chan M$. It follows that
\begin{equation}
\gcchan{M} \sim (\gcchan{R} \otimes \id) \comp \chan V_{\chan N}.
\end{equation}
But then
\begin{equation}
\gcchan{N} \sim (\tr \otimes \id) \comp \chan V_{\chan N} = (\tr \comp \gcchan{R} \otimes \id) \comp \chan V_{\chan N} \sim(\tr \otimes \id) \comp \gcchan{M},
\end{equation}
which implies that $\gcchan{M} \lnt \gcchan{N}$.
\qed


\subsection{Exact error correction}
\label{qec}

Theorem \ref{fund} directly yields the Knill-Laflamme conditions for exact quantum error correction on standard (subspace) codes. Suppose that the ``noise'' on the physical Hilbert space $\Hil$ is modeled by a channel with Kraus operators $E_i$. A standard code can be seen as a subspace of $\Hil$, or, equivalently, as an isometric encoding map $V: \Hil_C \hookrightarrow \Hil$ from the logical space $\Hil_C$ to the physical space $\Hil$. The projector $V V^\dagger = P$ projects on a subspace of $\Hil$ isomorphic to $\Hil_C$, which is usually called itself the ``code''.
The channel to consider then is
\begin{equation}
\chan N(\rho) = \sum_i E_i V \rho V^\dagger E_i^\dagger
\end{equation}
from $\Hil_C$ to $\Hil$. It implements the encoding followed by the noise. Given the code $P$ (or equivalently the isometric encoding $V$), the problem of quantum error correction is to find a {\em correction channel} $\chan R$ mapping $\Hil$ back to $\Hil_C$ such that
\begin{equation}
\label{standard_qec}
\chan R \chan N = \id_C,
\end{equation}
where $\id_C$ is the identity map on the logical space $\Hil_C$.

Hence, we can directly apply Theorem \ref{fund} with $\chan M = \id_C$. It is easy to see that all channels complementary to $\id_C$ are of the form
\begin{equation}
\rho \mapsto \proj{\psi} \, \tr(\rho)
\end{equation}
for an arbitrary pure state $\ket{\psi}$ living in a space of arbitrary dimension. In addition, one can show that all channels similar to such complementary channels are of the form
\begin{equation}
\rho \mapsto \tau \, \tr(\rho)\label{gencompid}
\end{equation}
for some mixed state $\tau$, i.e., Eq.~\eqref{gencompid} describes the generalized complementary channels to $\id_C$. Picking any of these channels, Theorem \ref{fund} states that the existence of a channel $\chan R$ satisfying Eq.~\eqref{standard_qec} is equivalent to the existence of a channel $\chan R'$ such that
\begin{equation}
\cchan N(\rho) = \chan R'(\tau) \, \tr(\rho).
\end{equation}
But since $\tau$ is fixed, this is equivalent to the existence of a state $\sigma$ such that
\begin{equation}
\cchan N(\rho) = \sigma \, \tr(\rho)
\end{equation}
for any state $\rho$. This is the Knill-Laflamme condition. This is most easily seen from its dual (Heisenberg picture) form, which reads
\begin{equation}
\cchanh N(A) = \tr( A \sigma ) \one
\end{equation}
for all operators $A$. It is enough to check this condition for a basis $\ketbra i j$ of the space of operators, which yields
\begin{equation}
\cchanh N(\ketbra i j ) = V^\dagger E_i^\dagger E_j V = \bra{j} \sigma \ket{i} \one.
\end{equation}
Multiplying on the left and on the right by $V$ and $V^\dagger$, respectively, we obtain
\begin{equation}
PE_i^\dagger E_j P = \lambda_{ij} P,
\end{equation}
where $\lambda_{ij} = \bra{j} \sigma \ket{i}$, which is the Knill-Laflamme condition in its most familiar form.

Moreover, we also obtain the generalized Knill-Laflamme conditions for the correctability of subsystem codes~\cite{kribs05,kribs06,poulin05x1}, or in fact any algebra, of which subsystem codes are a special case~\cite{beny07x1}.
A $\dagger$-algebra  (or algebra for short) is a set of operators closed under multiplication and which also contains the adjoint of all its elements. For instance, suppose that our Hilbert space $\Hil$ is divided into two subsystems: $\Hil = \Hil_A \otimes \Hil_B$. Then consider the set $\mathcal A$ of operators of the form $A \otimes \one$, where $A$ is an operator on $\Hil_A$ and $\one$ the identity on $\Hil_B$. It is trivial to show that $\mathcal A$ is an algebra. It represents all the local observables acting on $\Hil_1$. In fact this is close to being the most general form of a $\dagger$-algebra. For any $\dagger$-algebra $\mathcal A$ we can find a decomposition of the Hilbert space into orthogonal subspaces $\Hil_i$ which are left invariant by all elements of the algebra. Furthermore, when restricted to any of these invariant subspaces, the algebra has precisely the form described in the above example. Hence, the algebra defines a set of subsystems living in a family of orthogonal subspaces. This means that any element $A \in \mathcal A$ is of the form
\begin{equation}
A = \sum_i A_i \otimes \one_i,
\end{equation}
where $A_i \otimes \one_i$ is an operator supported on $\Hil_i$. Said differently, if $P_i$ is the projector on $\Hil_i$, then $P_i A P_i = A_i \otimes \one_i$.

A handy tool is the projector $\chan P_{\mathcal A}$ on this algebra, which we take to be orthogonal in terms of the Hilbert-Schmidt inner product between operators. This is a quantum channel satisfying $\chan P_{\mathcal A}^2 = \chan P_{\mathcal A} = \chan P_{\mathcal A}^\dagger$, whose range is precisely $\mathcal A$. It has the following explicit form:
\begin{equation}
\label{proj}
\chan P_{\mathcal A}(\rho) = \sum_i \tr_2(P_i \rho P_i) \otimes \frac{\one_i}{m_i},
\end{equation}
where $\tr_2$ is the partial trace over the second subsystem of the $i$th subspace, and $m_i$ is the dimension of that subsystem.

We say that an algebra $\mathcal A$ is correctable for the channel ${\chan N}$ if there exists a ``correction'' channel $\mathcal R$ such that for all $A \in \mathcal A$,
\begin{equation}
\label{ecd}
(\mathcal R \circ {\chan N})^\dagger(A) = A.
\end{equation}


Note that $\mathcal A$ contains the spectral projectors of any observable $A \in \mathcal A$. Hence, this definition implies that measuring $A$ before the action of the channel ${\chan N}$ or after the correction will yield the same probabilities, no matter what the initial state was.

Clearly, Eq.~\eqref{ecd} implies that $\chan P_{\mathcal A}\circ \mathcal R \circ {\chan N} = \chan P_{\mathcal A}$.
Hence, an equivalent formulation is to require the existence of a (possibly different) channel $\mathcal R$ such that
\begin{equation}
\label{ecd2}
\mathcal R \circ {\chan N} = \chan P_{\mathcal A}.
\end{equation}

This puts the problem in a form suitable for the application of Theorem~\ref{fund}, since it says that $\chan N \lnt \chan P_{\mathcal A}$.
If $\mathcal A'$ denotes the algebra formed by all operators commuting with all elements of $\mathcal A$ (i.e., the {\em commutant} of $\mathcal A$), one can show that the channel $P_{\mathcal A'}$ is a generalized complementary channel of $P_{\mathcal A}$. Elements of the commutant have the form $B = \sum_i \one_i \otimes B_i$ for the same decomposition into orthogonal subspaces. Explicitly,
\begin{equation}
\chan P_{\mathcal A'}(\rho) = \sum_i  \frac{\one_i}{n_i} \otimes \tr_1(P_i \rho P_i),
\end{equation}
where $\tr_1$ is the partial trace over the first subsystem of the $i$th subspace, and $n_i$ is the dimension of that subsystem.

Hence, Theorem~\ref{fund} states that the existence of $\mathcal R$ that satisfies Eq.~\eqref{ecd2} is equivalent to the existence of $\mathcal R'$ such that
\begin{equation}
\cchan N = \mathcal R' \chan P_{\mathcal A'}.
\end{equation}
But since $\chan P_{\mathcal A'}$ is a projector, this is equivalent to requiring that
\begin{equation}
\cchan N = \cchan N \chan P_{\mathcal A'},
\end{equation}
or, in the Heisenberg picture,
\begin{equation}
\cchanh N = \chan P_{\mathcal A'} \cchanh N.
\end{equation}
The latter is equivalent to requiring that for all operators $A$,
\begin{equation}
\cchanh N(A) \in \mathcal A',
\end{equation}
or, for all $i, j$, assuming $\chan N$ has the form previously used,
\begin{equation}
\cchanh N(\ketbra i j) = V^\dagger E_i^\dagger E_j V \in \mathcal A'.
\end{equation}
This is the form of the conditions derived, using a different method, in Ref.~\cite{beny07x1}.

\section{Main theorem}
\label{main-thm-section}

Let $f(\rho,\sigma) = \tr\sqrt{\sqrt \rho \, \sigma \sqrt \rho}$ be the fidelity \cite{uhlmann76} between states $\rho$ and $\sigma$. For reasons that will become clear, we extend the definition of this function to all positive operators $\rho$ or $\sigma$ of trace smaller than or equal to one (note that for operators of trace smaller than one, $f(\rho,\sigma)$ does not have the meaning of fidelity). Even for this more general concept, Uhlmann's theorem \cite{uhlmann76} holds, i.e, we have the alternative expression
\begin{equation}
\begin{split}
f(\rho,\sigma) &= \max_{V} |\bra{\psi_\rho} (\one \otimes V) \ket{\psi_\sigma}|\\
\end{split},
\end{equation}
where $\ket{\psi_\rho}$ and $\ket{\psi_\sigma}$ are any purifications of $\rho$ and $\sigma$, respectively, and the maximization runs over all isometric operators $V$ between the extra reference systems. (Note that here either $V^\dagger V = \one$ or $V V^\dagger =\one$ depending on which of the two purifications is of larger dimension).
Since the quantity $f(\rho,\sigma)$ is real, we can optimize the real part rather then the absolute value of the expression $\bra{\psi_\rho} (\one \otimes V) \ket{\psi_\sigma}$. In addition, one can show that the optimization can be done over all operators of norm smaller than one rather than just the isometric operators, i.e.,
\begin{equation}
\label{fidelity}
\begin{split}
f(\rho,\sigma) &= \max_{\|A\| \le 1} \re \bra{\psi_\rho} (\one \otimes A) \ket{\psi_\sigma}\\
&= \max_{\|A\| \le 1} \re
\xy (0,0)*{ \xy (-2,2)*{\psi_\sigma};(0,-2)*{};(0,6)*{}**\crv{(-6,-2)&(-6,6)};(0,-2)*{};(0,6)*{}**\dir{-};(0,0)*{};(4,0)*{};**\dir{-};(4,-2)*{};(8,-2)*{}**\dir{-};(8,-2)*{};(8,2)*{}**\dir{-};(8,2)*{};(4,2)*{}**\dir{-};(4,2)*{};(4,-2)*{}**\dir{-};(6,0)*{A};(8,0)*{};(12,0)*{};**\dir{-};{\ar(0,4)*{};(7,4)*{}};(7,4)*{};(12,4)*{};**\dir{-};(14,2)*{\psi_\rho};(12,-2)*{};(12,6)*{}**\crv{(18,-2) & (18,6)};(12,-2)*{};(12,6)*{}**\dir{-}; \endxy }; \endxy .
\end{split}
\end{equation}
The diagram can be thought of as a circuit where the boxes are operators which are not necessarily unitary, nor even square matrices. The left half circles represent input states, while the right half circles represent states which are scalar-multiplied with the corresponding ouputs. The above diagram thus represents the scalar product between the output of the circuit, $(\one \otimes A) \ket{\psi_\sigma}$, and the state $\ket{\psi_\rho}$.

For a given state $\rho$, we introduce the ``entanglement fidelity'' between channels $\chan N$ and $\chan M$,
\begin{equation*}
\begin{split}
F_\rho(\chan N, \chan M) &= f((\chan N \otimes \id)(\proj{\psi}), (\chan M \otimes \id)(\proj{\psi}))\\
&= \max_{\|A\| \le 1} \re
\xy (0,0)*{ \xy (-2,6)*{\psi};(0,2)*{};(0,10)*{}**\crv{(-6,2)&(-6,10)};(0,2)*{};(0,10)*{}**\dir{-};(0,4)*{};(4,4)*{};**\dir{-};(4,-2)*{};(12,-2)*{}**\dir{-};(12,-2)*{};(12,6)*{}**\dir{-};(12,6)*{};(4,6)*{}**\dir{-};(4,6)*{};(4,-2)*{}**\dir{-};(8,2)*{V_{\mathcal N}};(12,0)*{};(16,0)*{};**\dir{-};(16,-2)*{};(20,-2)*{}**\dir{-};(20,-2)*{};(20,2)*{}**\dir{-};(20,2)*{};(16,2)*{}**\dir{-};(16,2)*{};(16,-2)*{}**\dir{-};(18,0)*{A};(20,0)*{};(24,0)*{};**\dir{-};{\ar(12,4)*{};(19,4)*{}};(19,4)*{};(24,4)*{};**\dir{-};(24,-2)*{};(32,-2)*{}**\dir{-};(32,-2)*{};(32,6)*{}**\dir{-};(32,6)*{};(24,6)*{}**\dir{-};(24,6)*{};(24,-2)*{}**\dir{-};(28,2)*{V_{\mathcal M}^\dagger};(32,4)*{};(36,4)*{};**\dir{-};{\ar(0,8)*{};(19,8)*{}};(19,8)*{};(36,8)*{};**\dir{-};(38,6)*{\psi};(36,2)*{};(36,10)*{}**\crv{(42,2) & (42,10)};(36,2)*{};(36,10)*{}**\dir{-}; \endxy }; \endxy
\end{split},
\end{equation*}
where $\ket{\psi}$ is a purification of $\rho$. When $\chan M = \id$, this quantity reduces to the square root of Schumacher's entanglement fidelity of $\chan N$ \cite{schumacher96}. We will compare channels using the worst-case entanglement fidelity,
\begin{equation}
F(\chan N, \chan M) = \min_\rho F_\rho(\chan N, \chan M),
\end{equation}
which was studied in Ref.~\cite{gilchrist05}.

\begin{theorem}
\label{main} If $\cchan N$ and $\cchan M$ are complementary to $\chan N$ and $\chan M$, respectively, then for any $d>1$,
\begin{equation}
\label{fidequ} \max_{|\chan R| \le d} F(\chan R \chan N, \chan M) = \max_{|\chan R'| \le d} F(\cchan N, \chan R' \cchan M),
\end{equation}
where the maxima are over all {\it trace-nonincreasing} completely positive maps with the appropriate source and target spaces, and $|\chan R|$ stands for the minimal number of Kraus operators for $\chan R$.
\end{theorem}
\begin{proof}
The proof closely follows arguments used in Ref.~\cite{kretschmann08}.
Let $V_{\chan N}$ be the isometry for which $\chan N(\rho) = \tr_E (V_{\chan N}\rho V_{\chan N}^\dagger)$ and $\cchan N(\rho) = \tr_B (V_{\chan N}\rho V_{\chan N}^\dagger)$, and $V_{\chan M}$ be the isometry yielding both $\chan M$ and $\cchan M$ in the same way. Note that any trace-nonincreasing channel $\chan R$ can be written as $\chan R(\rho) = \tr_{\widetilde E}(A \rho A^\dagger)$ for some operator $A$ satisfying $\|A\| \le 1$ from the input Hilbert space of $\chan R$ to its output space tensored with an ``environment'' $\widetilde E$. 
Using this fact and writing the fidelity using Eq.~\eqref{fidelity}, we obtain
\begin{gather}
\label{minimax1}
\max_{|\chan R|\le d} F(\chan R \chan N, \chan M) 
=\max_{\| A \| \le 1} \min_{\rho} \max_{\| A' \| \le 1} \re \, g_\rho(A,A'),
\end{gather}
where $g_\rho$ can be expressed in terms of a circuit:
\begin{equation}
\label{maindiag}
g_\rho(A,A') =
\xy(0,0)*+{\xy(-1.8,3)*{\psi};(0,-1.5)*{};(0,7.5)*{}**\crv{(-6,-1.5)&(-6,7.5)};(0,-1.5)*{};(0,7.5)*{}**\dir{-};(0,0)*{};(3,0)*{}**\dir{-};(3,-7.5)*{};(12,-7.5)*{}**\dir{-};(12,-7.5)*{};(12,1.5)*{}**\dir{-};(12,1.5)*{};(3,1.5)*{}**\dir{-};(3,1.5)*{};(3,-7.5)*{}**\dir{-};(7.5,-3)*{V_{\mathcal N}};(12,0)*{};(15,0)*{}**\dir{-};(13.5,2)*{\scriptstyle B};(15,-4.5)*{};(21,-4.5)*{}**\dir{-};(21,-4.5)*{};(21,1.5)*{}**\dir{-};(21,1.5)*{};(15,1.5)*{}**\dir{-};(15,1.5)*{};(15,-4.5)*{}**\dir{-};(18,-1.5)*{A};(21,-3)*{};(24,-3)*{}**\dir{-};(22.5,-1)*{\scriptstyle d};(24,-6)*{};(19,-6)*{}**\dir{-};{\ar(12,-6)*{};(19,-6)*{}};(18,-8)*{\scriptstyle E};(24,-7.5)*{};(30,-7.5)*{}**\dir{-};(30,-7.5)*{};(30,-1.5)*{}**\dir{-};(30,-1.5)*{};(24,-1.5)*{}**\dir{-};(24,-1.5)*{};(24,-7.5)*{}**\dir{-};(27,-4.5)*{{{A'}^\dagger}};(30,-6)*{};(33,-6)*{}**\dir{-};(31.5,-8)*{\scriptstyle E'};(33,0)*{};(28,0)*{}**\dir{-};{\ar(21,0)*{};(28,0)*{}};(27,2)*{\scriptstyle B'};(33,-7.5)*{};(42,-7.5)*{}**\dir{-};(42,-7.5)*{};(42,1.5)*{}**\dir{-};(42,1.5)*{};(33,1.5)*{}**\dir{-};(33,1.5)*{};(33,-7.5)*{}**\dir{-};(37.5,-3)*{V_{\mathcal M}^\dagger};(42,0)*{};(45,0)*{}**\dir{-};(45,6)*{};(23.5,6)*{}**\dir{-};{\ar(0,6)*{};(23.5,6)*{}};(46.8,3)*{\psi};(45,-1.5)*{};(45,7.5)*{}**\crv{(51,-1.5)&(51,7.5)};(45,-1.5)*{};(45,7.5)*{}**\dir{-};\endxy};\endxy
\end{equation}
The small $d$ indicates that the wire below it represents a Hilbert space of dimension $d$, namely, the system $\widetilde E$ mentioned above.
The wires labeled $B$ and $B'$ represent the target systems for $\chan N$ and $\chan M$, respectively, and $E$ and $E'$ are the respective ``environments''. The state $\ket{\psi_\rho}$ can be any purification of $\rho$. If we reflect the picture with respect to a vertical axis through the middle, Hermitian conjugating each operator [this amounts to a complex conjugation of $g_{\rho}(A,A')$], and exchange the wire labels $E'$ and $B$, and $E$ and $B'$, we see that we also have
\begin{equation}
\max_{|\chan R'|\le d} F( \cchan N, \chan R' \cchan M) =\max_{\|A'\| \le 1} \min_{\rho} \max_{\|A\| \le 1} \re\, g_\rho(A,A'),
\end{equation}
where now $A'$ is the operator defining $\chan R'$ while $A$ comes from Eq.~\eqref{fidelity} for the fidelity. Hence, we just have to show that we can exchange the maximizations over $A$ and $A'$ in Eq.~\eqref{minimax1}.
This can be done by applying Shiffman's minimax theorem~\cite{grossinho01} which says that we can exchange the rightmost min and max provided that the function is convex-concave in the two arguments (in this case it is bilinear), and that the variables are optimized over convex sets, which is the case here. Hence, we obtain $\max_{\chan R} F(\chan R \chan N, \chan M) =\max_{A'}\max_{A} \min_{\rho} \re g_\rho(A,A')= \max_{\chan R'} F( \cchan N, \chan R' \cchan M)$, where $\|A'\|, \|A\|\le 1$.
\end{proof}

\begin{proposition}
\label{robustprop}
If $\chan N' \sim \chan N$ and $\chan M' \sim \chan M$, then
\begin{equation}
\label{robust} \max_{\chan R} F(\chan R \chan N, \chan M) = \max_{\chan R} F(\chan R \chan N', \chan M'),
\end{equation}
where the maxima are over all quantum channels with the appropriate source and target spaces.
\end{proposition}
\begin{proof}
We have $\chan N' = \chan S \chan N$, $\chan N = \chan S' \chan N'$, $\chan N' = \chan T \chan N$, $\chan N = \chan T' \chan N'$ for some channels $\chan S, \chan S', \chan T, \chan T'$. Hence,
\begin{equation}
\begin{split}
 \max_{\chan R} F(\chan N, \chan R \chan M) &= F(\chan N, \chan R_0 \chan M)\\
&\le F(\chan S \chan N, \chan S \chan R_0 \chan M)\\
&= F(\chan N', \chan S \chan R_0 \chan T' \chan M')\\
&\le \max_{\chan R} F(\chan N', \chan R\chan M').\\
\end{split}
\end{equation}
The converse inequality follows in the same way.
\end{proof}
\proof
Suppose that $\chan N$ maps operators over $\Hil_A$ to operators over $\Hil_B$, and $\chan M$ maps operators over $\Hil_A$ to operators over $\Hil_C$. Then,
using Proposition~\ref{robustprop} together with the main theorem for $d$ maximal, i.e.,
\begin{equation}
d = \dim(\Hil_A)^2 \dim(\Hil_B)\dim(\Hil_C),
\end{equation}
we obtain a variation of Theorem~\ref{main} which we will find most useful, and which is the direct generalization of Theorem~\ref{fund}:
\begin{corollary}
\label{main2}
If $\gcchan{N}$ and $\gcchan{M}$ are generalized complementary channels of $\chan N$ and $\chan M$, respectively, then
\begin{equation}
\label{fidequ2} \max_{\chan R} F(\chan R \chan N, \chan M) = \max_{\chan R'} F(\gcchan{N}, \chan R' \gcchan{M}),
\end{equation}
where the maxima are over quantum channels with the appropriate source and target spaces.
\end{corollary}
\proof  
The only thing we have to show is that the maxima taken without constraint on the number of Kraus operators are always attained by trace-preserving maps, i.e. quantum channels.
We use the fact that for any states $\rho$,$\tau$, and $\sigma$, and $0 \le p \le 1$, we have from the strong concavity of the fidelity \cite{nielsen00} that
\begin{equation}
f(p \rho + (1-p) \tau , \sigma) \ge \sqrt p\, f(\rho, \sigma) = f(p \rho, \sigma).
\end{equation}
Suppose that $\chan R$ is a trace-nonincreasing completely positive map. We can always ``complete'' it to a trace-preserving channel
\(
\overline{\chan R} = \chan R + \chan S,
\)
where $\chan S$ is another completely positive map. For example, one can take $\chan S(\rho)=\tr(\rho-\chan R(\rho))\tau$ for some state $\tau$.
We then have, using the shorthands $\chan N^e \equiv \chan N \otimes \id$ and $\psi \equiv \proj{\psi}$,
\begin{equation}
\begin{split}
F(\overline{\chan R} \chan N, \chan M) &= \min_\psi f[(\overline {\chan R}\chan N)^e(\psi), \chan M^e(\psi)]\\
&= f[(\overline {\chan R}\chan N)^e({\psi_0}), \chan M^e({\psi_0})]\\
&= f[(\chan R\chan N)^e(\psi_0) + (\chan S\chan N)^e({\psi_0}), \chan M^e(\psi_0)]\\
&\ge f[(\chan R\chan N)^e(\psi_0), \chan M^e(\psi_0)]\\
&\ge \min_\psi f[(\chan R\chan N)^e(\psi), \chan M^e(\psi)]\\
&= F(\chan R \chan N, \chan M).\\
\end{split}
\end{equation}
The same argument works for the right-hand side of Eq.~\eqref{fidequ2}.
\qed


\section{Special case $\gcchan M^2 = \gcchan M$}
\label{special-case-section}

Theorem~\ref{main} and Corollary~\ref{main2} might not seem directly useful since they express one optimization in terms of a different but seemingly equally hard one. However, we will show that there are interesting problems, in particular error correction and minimax state discrimination, where one of the optimizations can be given a general and straightforward near-optimal solution. More generally, we will consider the case where \begin{equation}
\gcchan M^2 = \gcchan M.
\end{equation}
One can readily check that this can indeed be satisfied in the special case where $\chan M = \one$, i.e., quantum error correction.

\subsection{Near-optimal bounds}
\label{special-case-bounds}

We will concentrate here on the form of our theorem given in Corollary~\ref{main2}, however, it is straightforward to apply the same reasoning to Theorem~\ref{main}.

In this section, we will replace the worst-case entanglement fidelity $F(\chan N, \chan M)$ by a {\em distance} $d(\chan N, \chan M)$.
From the Bures distance \cite{bures69}, we can define
\(
d_\rho(\chan N, \chan M) = \sqrt{1 - {F_\rho(\chan N, \chan M)}},
\)
which can be used to define
\begin{equation}
d(\chan N, \chan M) := \max_\rho d_\rho(\chan N, \chan M) = \sqrt{1 - {F(\chan N, \chan M)}}
\end{equation}
which satisfies the triangle inequality:
\begin{equation}
\begin{split}
d(\chan N, \chan M) &= \max_\rho d_\rho(\chan N, \chan M) \\
&\le \max_\rho [ d_\rho(\chan N, \chan R) + d_\rho(\chan R, \chan M) ] \\
&\le d(\chan N, \chan R) + d(\chan R, \chan M).
\end{split}
\end{equation}
The equation in Corollary~\ref{main2} in terms of this distance becomes
\begin{equation}
\min_{\chan R} d(\chan R \chan N, \chan M) = \min_{\chan R'} d(\gcchan{N}, \chan R' \gcchan{M}).
\end{equation}


\begin{corollary}
If $\gcchan{N}$ and $\gcchan{M}$ are generalized complementary channels of $\chan N$ and $\chan M$, respectively, and $\gcchan{M}^2 = \gcchan{M}$, then
\begin{equation}
\label{bounds}
\frac{1}{2}d(\gcchan{N}, \gcchan{N} \gcchan{M}) \le \min_{\chan R} d(\chan R \chan N, \chan M) \le d(\gcchan{N}, \gcchan{N} \gcchan{M}).
\end{equation}
\end{corollary}
\proof
The rightmost inequality follows from picking the suboptimal $\chan R' = \gcchan{N}$. For the leftmost inequality, suppose that $\chan R'_0$ minimizes $d(\gcchan{N}, \chan R' \gcchan{M})$. Then, using the triangle inequality,
\begin{equation}
d(\gcchan{N}, \gcchan{N} \gcchan{M}) \le d(\gcchan{N}, \chan R'_0 \gcchan{M}) + d(\chan R'_0 \gcchan{M}, \gcchan{N} \gcchan{M}).
\end{equation}
Let
\(
\varepsilon_0 := \min_{\chan R} d(\chan R \chan N, \chan M).
\)
We know that the first term is equal to $\varepsilon_0$ since $\chan R'_0$ is optimal. For the second term, note that
\begin{equation}
 d(\chan R'_0 \gcchan{M}, \gcchan{N} \gcchan{M}) = d(\chan R'_0 \gcchan{M}^2, \gcchan{N} \gcchan{M}) \le d(\chan R'_0\gcchan{M}, \gcchan{N})  = \varepsilon_0.
\end{equation}
For the last inequality, we used the fact that
\begin{equation}
\label{rmonot}
d(\chan N \chan R,\chan M \chan R) \le d(\chan N, \chan M)
\end{equation}
for any channels $\chan N$, $\chan M$, and $\chan R$. This property follows from the fact that $\chan R$ simply limits the number of input states over which the maximum is taken inside the definition of the distance.
It follows that $d(\gcchan{N}, \gcchan{N} \gcchan{M}) \le 2\varepsilon_0$.
\qed

Note that computing $d(\cchan N, \cchan N \cchan M)$ requires a convex maximization over inputs only \cite{gilchrist05}, which is a significant simplification over the minimax $\min_{\chan R} d(\chan R\chan N,\chan M)$.

\subsection{Near-optimal recovery channels}
\label{special-case-channel}

Let us show how we can construct a recovery channel $\good{\chan R}$ which performs as well as guaranteed by our bounds (Eq.~\eqref{bounds}), i.e.,
\begin{equation}
\label{goodchan}
d(\good{\chan R} \chan N, \chan M) \le d(\gcchan{N}, \gcchan{N}\gcchan{M}).
\end{equation}

If we take $\gcchan{M}=\cchan M$ to be complementary to $\chan M$ and $\gcchan{N}=\cchan N$ to be complementary to $\chan N$, then the fidelity $F(\cchan N, \cchan N\cchan M) = 1 - d(\cchan N, \cchan N\cchan M)^2$ is given through a minimization and a maximization as
\begin{equation}
F(\cchan N, \cchan N\cchan M) = \min_\rho \max_{\|A\| \le 1} \re\, g_{\rho}(A,U'),
\end{equation}
where the bilinear function $g_\rho$ is defined in Eq.~\eqref{maindiag}, and $U'$ yields $\cchan N$ through $\cchan N(\rho) = \tr_2[U'(\rho \otimes \proj{0}) (U')^\dagger]$.

The minimax theorem guarantees that there exists a {\em saddle point}, i.e., a pair $(\rho_0,A_0)$ such that we have both
\begin{equation}
\re\, g_{\rho_0}(A_0,U') = \min_\rho \max_{\|A\| \le 1} \re\, g_{\rho}(A,U')
\end{equation}
and
\begin{equation}
\re\, g_{\rho_0}(A_0,U') =  \max_{\|A\| \le 1} \min_\rho \re\, g_{\rho}(A,U').
\end{equation}
If we know this saddle point, then, defining the trace-nonincreasing completely positive map
\begin{equation}
\chan S(\rho) := \tr_2(A_0(\rho \otimes \proj{0}) A_0^\dagger)
\end{equation}
and completing it to the trace-preserving channel
\begin{equation}
\good{\chan R}(\rho) = \chan S(\rho) + \tr(\rho - \chan S(\rho)) \tau
\end{equation}
for some state $\tau$,
we have
\begin{equation}
\begin{split}
F(\good{\chan R} \chan N, \chan M) &\ge F(\chan S \chan N, \chan M)\\
&= \min_\rho \max_{U'} \re\, g_{\rho}(A_0,U')\\
&= \max_{\|A'\|\le 1} \min_\rho \re\, g_{\rho}(A_0,A')\\
&\ge \min_\rho \re\, g_{\rho}(A_0,U')\\
&= \re\, g_{\rho_0}(A_0,U')\\
&= F(\cchan N, \cchan N\cchan M),\\
\end{split}
\end{equation}
i.e., $\good{\chan R}$ satisfies Eq.~\eqref{goodchan}.
Hence, a near-optimal correction channel $\good{\chan R}$ is given by the saddle point in the minimax problem yielding the estimate $F(\cchan N, \cchan N \cchan M)$.
For completeness, suppose that instead of $\cchan M$ complementary to $\chan M$, we use a generalized complementary channel $\gcchan{M} \sim \cchan M$. Let $\chan M'$ be complementary to $\gcchan{M}$. Using Theorem~\ref{fund} we obtain $\chan M' \sim \chan M$. As above, we can build $\good{\chan R}'$ such that
\begin{equation}
d(\good{\chan R}' \chan N, \chan M') \le d(\cchan N, \cchan N\gcchan{M}).
\end{equation}
Suppose that $\chan T'$ is such that $\chan M = \chan T' \chan M'$. Then using $\good{\chan R} := \chan T' \good{\chan R}'$, we obtain
\begin{equation}
d(\good{\chan R} \chan N, \chan M) \le d(\good{\chan R}' \chan N, \chan M') \le d(\cchan N, \cchan N\gcchan{M}).
\end{equation}
If furthermore $\gcchan{N} \sim \cchan N$, we have
\begin{equation}
d(\good{\chan R} \chan N, \chan M) \le d(\cchan N, \cchan N\gcchan{M}) \le d(\gcchan{N}, \gcchan{N} \gcchan{M})
\end{equation}
by the monotonicity of the distance.

Let us now focus on the problem of finding the saddle point in the case where
\begin{equation}
\chan M(\rho) = \rho \otimes \sigma
\end{equation}
and we use the complementary channel
\begin{equation}
\cchan M(\rho) = \sigma \tr(\rho).
\end{equation}
We use a channel $\chan M$ slightly more general than for pure quantum error correction so that we can use for $\cchan M$ the most general channel similar to a channel complementary to the identity. This means that to simulate the identity rather than this channel $\chan M$, we can just use the near-optimal channel $\good{\chan R}$ that we will obtain and trace out the extra state $\sigma$.

We also write
\begin{equation}
\label{encnoise}
\chan N(\rho) = \sum_i E_i V \rho V^\dagger E_i^\dagger
\end{equation}
as in Section~\ref{qec}.

Assuming $\sigma = \sum_j p_j \proj{j}$, and writing a purification of $\sigma$ as $\ket{\psi} = \sum_i \sqrt {p_i} \ket{i}_A \otimes \ket{i}_R$, we define the operator
\begin{equation}
\begin{split}
X_\rho &:=
\xy (0,0)*{ \xy (2,-2)*{\psi};(4,-6)*{};(4,2)*{}**\crv{(-2,-6)&(-2,2)};(4,-6)*{};(4,2)*{}**\dir{-};(4,0)*{};(8,0)*{};**\dir{-};(8,-2)*{};(16,-2)*{}**\dir{-};(16,-2)*{};(16,6)*{}**\dir{-};(16,6)*{};(8,6)*{}**\dir{-};(8,6)*{};(8,-2)*{}**\dir{-};(12,2)*{V_{\mathcal N}};(16,4)*{};(20,4)*{};**\dir{-};{\ar(0,8)*{};(11,8)*{}};(11,8)*{};(20,8)*{};**\dir{-};(20,2)*{};(28,2)*{}**\dir{-};(28,2)*{};(28,10)*{}**\dir{-};(28,10)*{};(20,10)*{}**\dir{-};(20,10)*{};(20,2)*{}**\dir{-};(24,6)*{V^\dagger_{\mathcal N}};(28,8)*{};(32,8)*{};**\dir{-};(32,6)*{};(36,6)*{}**\dir{-};(36,6)*{};(36,10)*{}**\dir{-};(36,10)*{};(32,10)*{}**\dir{-};(32,10)*{};(32,6)*{}**\dir{-};(34,8)*{\rho};(36,8)*{};(40,8)*{};**\dir{-};(38,10)*{\scriptstyle A};{\ar(16,0)*{};(29,0)*{}};(29,0)*{};(40,0)*{};**\dir{-};(28,-2)*{\scriptstyle B};{\ar(4,-4)*{};(23,-4)*{}};(23,-4)*{};(40,-4)*{};**\dir{-};(22,-6)*{\scriptstyle R}; \endxy }; \endxy \\
&= \sum_i
\xy (0,0)*{ \xy (2,-2)*{\psi};(4,-6)*{};(4,2)*{}**\crv{(-2,-6)&(-2,2)};(4,-6)*{};(4,2)*{}**\dir{-};(4,0)*{};(8,0)*{};**\dir{-};(8,-2)*{};(16,-2)*{}**\dir{-};(16,-2)*{};(16,6)*{}**\dir{-};(16,6)*{};(8,6)*{}**\dir{-};(8,6)*{};(8,-2)*{}**\dir{-};(12,2)*{V_{\mathcal N}};(16,4)*{};(20,4)*{};**\dir{-};(16,4)*{};(20,4)*{};**\dir{-};(21.6,4)*{\scriptstyle i};(20,2)*{};(20,6)*{}**\crv{(24,2) & (24,6)};(20,2)*{};(20,6)*{}**\dir{-};(26.4,4)*{\scriptstyle i};(28,2)*{};(28,6)*{}**\crv{(24,2)&(24,6)};(28,2)*{};(28,6)*{}**\dir{-};(28,4)*{};(32,4)*{};**\dir{-};{\ar(0,8)*{};(17,8)*{}};(17,8)*{};(32,8)*{};**\dir{-};(32,2)*{};(40,2)*{}**\dir{-};(40,2)*{};(40,10)*{}**\dir{-};(40,10)*{};(32,10)*{}**\dir{-};(32,10)*{};(32,2)*{}**\dir{-};(36,6)*{V^\dagger_{\mathcal N}};(40,8)*{};(44,8)*{};**\dir{-};(44,6)*{};(48,6)*{}**\dir{-};(48,6)*{};(48,10)*{}**\dir{-};(48,10)*{};(44,10)*{}**\dir{-};(44,10)*{};(44,6)*{}**\dir{-};(46,8)*{\rho};(48,8)*{};(52,8)*{};**\dir{-};(50,10)*{\scriptstyle A};{\ar(16,0)*{};(35,0)*{}};(35,0)*{};(52,0)*{};**\dir{-};(34,-2)*{\scriptstyle B};{\ar(4,-4)*{};(29,-4)*{}};(29,-4)*{};(52,-4)*{};**\dir{-};(28,-6)*{\scriptstyle R}; \endxy }; \endxy \\
&= \sum_i \rho V^\dagger E_i^\dagger \otimes (E_i V \otimes \one )\ket{\psi}\\
&= \sum_{ij} \rho V^\dagger E_i^\dagger \otimes \sqrt {p_j}   E_iV  \ket j \otimes \ket j.
\end{split}
\end{equation}
(We note that the operator $X_\rho$ defined here is different from the one denoted by the same symbol in Ref.~\cite{beny10}.)
One can check that
\begin{equation}
g_\rho(A,U') = \tr[A (X_\rho^\dagger \otimes \ket 0)],
\end{equation}
where the state $\ket{0}$ is the state used before to relate $A$ to a completely positive map $\chan S$. Since this state is arbitrary, we can just absorb it in the definition of a new operator $A$, which is now not given by a square matrix anymore, and we write simply
\begin{equation}
g_\rho(A,U') = \tr(A X_\rho^\dagger),
\end{equation}
where the completely positive map $\chan S$ is obtained from $A$ by
\begin{equation}
\chan S(\rho) := \tr_B(A \rho A^\dagger).
\end{equation}
In terms of $X_\rho$, we then have
\begin{gather}
F(\cchan N, \cchan N \cchan M) = \min_\rho \max_{\|A\|\le 1} \re \tr(A X_\rho^\dagger) =  \min_\rho \tr |X_\rho|,\label{intermsofX}
\end{gather}
where
\begin{equation}
|X_\rho| := \sqrt{X_\rho^\dagger X_\rho}.
\end{equation}
The second equality in Eq.~\eqref{intermsofX} uses the fact that $\max_{\|A\|\le 1} \re \tr(A X_\rho^\dagger) =  \tr |X_\rho|$, which can be seen as follows. Let $X_\rho=V|X_\rho|$ be a polar decomposition of $X_\rho$. Then $\re \tr(A X_\rho^\dagger)\equiv  \re \tr(A |X_\rho|V^\dagger)=\re\tr(V^\dagger A |X_\rho|)$. Since $\|V^\dagger A \|\leq \|V^\dagger\|\| A \|\leq 1$ and $|X_\rho|\geq 0$, we have (see Appendix \ref{firstapp}) $\re\tr(V^\dagger A |X_\rho|)\leq \tr|X_\rho|$. This bound is achievable for $A=V$.

The following result shows that under certain circumstances one can use as a saddle point the state $\rho_0$ minimizing $\tr |X_\rho|$ together with a unitary operator $A_0$ coming from the polar decomposition of $X_{\rho_0}$,
\begin{equation}
X_{\rho_0}=A_0 |X_{\rho_0}|.
\end{equation}
To express $\chan S$, it will be convenient to introduce the completely positive map $\Phi_\rho$ defined by
\begin{equation}
\label{thatmap}
\begin{split}
\Phi_\rho^\dagger(\tau) &= \tr_{BR}( X_\rho \tau X_\rho^\dagger )\\
&= \sum_{ij} \tr(\sigma V^\dagger E_i^\dagger E_j V) \rho V^\dagger E_j^\dagger \tau E_i V \rho.
\end{split}
\end{equation}
Note also that
\begin{equation}
X_\rho^\dagger X_\rho = \Phi_\rho(\one).
\end{equation}

\begin{proposition}
Let $\chan N$ be defined as in Eq.~\eqref{encnoise} and $\Phi$ as in Eq.~\eqref{thatmap}. Let $\rho_0$ be a state minimizing
\begin{equation}
F_0(\rho) := \tr\sqrt{\Phi_\rho(\one)}.
\end{equation}
We have
\begin{equation}
F_0(\rho_0)  \le \max_{\chan R} F(\chan R \chan N, \id) \le \frac 1 4 F_0(\rho_0) + \frac 3 4.
\end{equation}
If furthermore the optimal state $\rho_0$ is {\em unique} and of {\em full rank}, then any channel of the form
\begin{equation}
\label{goodchannel}
\good{\chan R}(\tau) = \Phi_{\rho_0}^\dagger[\Phi_{\rho_0}(\one)^{-\frac 1 2} \tau \Phi_{\rho_0}(\one)^{-\frac 1 2}] + \chan T(\tau),
\end{equation}
for some completely positive map $\chan T$,
is near-optimal in the sense that it satisfies
\begin{equation}
F_0(\rho_0) \le F(\good{\chan R} \chan N, \id).
\end{equation}
\end{proposition}
\proof
The inequalities involving $F_0(\rho_0)$ follow directly from Proposition~\ref{bounds}.
In view of the previous discussion, we only need to prove that the optimal state $\rho_0$ together with an operator $A_0$ solving $X_{\rho_0}=A_0 |X_{\rho_0}|$ form a saddle point for $\re \tr(A X^\dagger_\rho)$. It is clear from the definitions that this pair $(A_0,\rho_0)$ attains the value $\min_\rho \max_{\|A\|=1} \re \tr(A X^\dagger_\rho)$. But we do not a priori know which solution $A_0$ of the equation $X_{\rho_0}=A |X_{\rho_0}|$ is such that $\rho_0$ is the minimum for $\re \tr(A_0 X^\dagger_\rho)$, which is what is needed for $(A_0,\rho_0)$ to be a saddle point. We show in Appendix \ref{obvious} that if $\rho_0$ has full rank, then for any state $\rho'$, $\re \tr(A_0 X_{\rho'})$ is independent of which solution $A_0$ of $X_{\rho_0}=A |X_{\rho_0}|$ we choose. Hence, any of them yields a near-optimal correction channel.
\qed

Note that the condition that $\rho_0$ be full-rank is necessary for the above construction to work. A counterexample to this construction for the case when $\rho_0$ is not of full rank is presented in Appendix \ref{counter}.

\subsection{Nature of the approximately correctable channels}
\label{nature}

It would seem natural that an approximately correctable channel is also close to some exactly correctable channel. Here we will show a slight variation of this intuition in terms of the worst-case entanglement fidelity, namely that an approximately correctable channel is always {\em similar} to one which is close to an exactly correctable channel.

\begin{theorem}
\label{iso1}
For any channel $\chan N$, there exists an exactly correctable channel $\chan N_0$ and a channel $\chan N' \sim \chan N$ such that
\begin{equation}
d(\chan N', \chan N_0) = \min_{\chan R} d(\chan R \chan N, \id).
\end{equation}
\end{theorem}
\proof
Let
\begin{equation}
\varepsilon := \min_{\chan R} d(\chan R \chan N, \id).
\end{equation}
To prove that $d(\chan N', \chan N_0) \ge \varepsilon$, consider the channel $\chan R_0$ correcting $\chan N_0$. Using first the monotonicity of the distance and then Corollary~\ref{robustprop}, we conclude
\begin{equation}
d(\chan N', \chan N_0) \ge d(\chan R_0 \chan N', \id) \ge \min_{\chan R} d(\chan R \chan N', \id) = \varepsilon.
\end{equation}

For the converse, observe that by Corollary~\ref{main2} there exists a constant channel $\chan C$,
\begin{equation}
\chan C(\rho) = \sigma \tr(\rho),
\end{equation}
such that
\begin{equation}
d(\cchan N, \chan C) = \varepsilon.
\end{equation}
Using Theorem~\ref{main} with $d=1$,
\begin{equation}
\varepsilon = d(\cchan N, \chan C) \ge \min_{\chan |R|=1} d(\chan R \cchan N, \chan C) = \min_{\chan |R|=1} d(\chan N, \chan R \cchan C).
\end{equation}
Let $R_0$ be the optimal trace-nonincreasing CP map on the right-hand side. It is of the form $\chan R_0(\rho) = A \rho A^\dagger$ where $A^\dagger A \le \one$.
Consider the isometry
\begin{equation}
V := A \otimes \ket 0 + \sqrt{\one - A^\dagger A} \otimes \ket 1
\end{equation}
with corresponding map $\chan V(\rho) := V \rho V^\dagger$.
We have
\begin{equation}
\label{interm}
\varepsilon \ge d(\chan N, \chan R_0 \cchan C) = d(\chan N', \chan V \cchan C)
\end{equation}
where
\begin{equation}
\chan N'(\rho) := \chan N(\rho) \otimes \proj 0.
\end{equation}
Clearly $\chan N' \sim \chan N$.
In addition, we know from the Knill-Laflamme conditions (see Section~\ref{qec}) that $\cchan C$ is exactly correctable. Therefore $\chan N_0 = \chan V \cchan C$ is also exactly correctable.
\qed

Let us find a channel $\cchan C$ complementary to $\chan C$ explicitly. If $\sigma = \sum_i p_i \proj{i}$, we have
\begin{equation}
\chan C(\rho) = \sum_{ij} p_i \ketbra i j \rho \ketbra j i = \tr_E W \rho W^\dagger,
\end{equation}
where
\begin{equation}
W = \sum_{ij} \sqrt{p_i} \ketbra i j \otimes \ket{ij}_E.
\end{equation}
The corresponding complementary channel is
\begin{equation}
\cchan C(\rho) = \tr_B W \rho W^\dagger = \sum_{ijj'} p_i \bra{j}\rho\ket{j'} \otimes \ket{ij}\bra{ij'} = \rho \otimes \sigma.
\end{equation}
Hence we see that the exactly correctable channel $\chan N_0$ in the above proof has the form
\begin{equation}
\chan N_0(\rho) = V (\rho \otimes \sigma) V^\dagger,
\end{equation}
which is indeed the general form of a correctable channel~\cite{nayak06}.

\subsection{State discrimination}
\label{discrimination}

To illustrate the generality of our result, let us show how it yields a nontrivial result for minimax state discrimination \cite{DAriano05}.

We want to specialize the relation
\begin{equation}
\frac{1}{2} d(\cchan N, \cchan N \cchan M) \le \min_{\chan R} d(\chan R \chan N, \chan M) \le d(\cchan N, \cchan N \cchan M),
\end{equation}
where $d(\chan N, \chan M) = \sqrt{1-F(\chan N,\chan M)}$ and $F(\chan N,\chan M)$ is the worst-case entanglement fidelity between channels $\chan N$ and $\chan M$,
to the case where
\begin{equation}
\chan M(\rho) = \sum_i \proj{i} \rho \proj{i}
\end{equation}
and
\begin{equation}
\chan N(\rho) = \sum_i \rho_i \bra{i} \rho \ket{i} = \sum_{ij} s_i^\dagger \ket{j}\bra{i} \rho \ket{i}\bra{j} s_i
\end{equation}
for a fixed set of states $\{\rho_i\}$, where $\rho_i = s_i^\dagger s_i$. This is the problem of minimax state discrimination \cite{DAriano05}. Indeed, this channel $\chan N$ can be thought of as a classical-to-quantum channel, which is just what a state preparation is: it outputs a quantum state depending on some classical data, namely, the choice of which $\rho_i$ to output. 

Since the output of $\chan M$ is diagonal in the basis $\ket{i}$, we also expect the output of the optimal channel $\chan R$ to be, in which case there is a positive operator-valued measure (POVM) with elements $A_i$ ($\sum_iA_i=\one$) such that
\begin{equation}
\chan R(\rho) = \sum_i \tr(\rho A_i) \proj{i}.
\end{equation}
We then have
\begin{equation}
\max_{\chan R} F(\chan R \chan N, \chan M) = \max_A \min_p \sum_i \sqrt{p_i} \sqrt{ \sum_j p_j  \tr( \rho_j A_i ) },
\end{equation}
where $p_i = \bra{i} \rho \ket{i}$ is all that matters about the initial state $\rho$.
Since the classical fidelity is jointly concave in both arguments, the minimum of $p$ is achieved when $p$ is pure. Hence,
\begin{equation}
\max_{\chan R} F(\chan R \chan N, \chan M) = \max_A \min_i \sqrt{ \tr( \rho_i A_i ) }.
\end{equation}
The square of this expression is the minimal worst-case success probability for the discrimination of the states $\rho_i$.

In order to see what the bound $F(\cchan N, \cchan N \cchan M)$ is, we need channels complementary to $\chan N$ and $\chan M$, or the corresponding isometries $V$ and $W$.
If
\begin{equation}
W = \sum_i \ket{i}_E \otimes \ket{i}_B \bra{i},
\end{equation}
then clearly $\chan M(\rho) = \tr_E W \rho W^\dagger$. Hence, we can pick $\cchan M := \tr_B W \rho W^\dagger$.
Similarly, if
\begin{equation}
V = \sum_{ij}  \ket{ij}_E \otimes s_i^\dagger \ket{j}_B\bra{i},
\end{equation}
then $\chan N(\rho) = \tr_E V \rho V^\dagger$, and so we can use $\cchan N := \tr_B V \rho V^\dagger$.
The worst-case entanglement fidelity $F(\cchan N, \cchan N \cchan M)$ then can be written as
\begin{equation}
F(\cchan N, \cchan N \cchan M) = \min_\psi \max_U |f(U,\rho)|,
\end{equation}
where
\begin{equation}
\begin{split}
f(U,\rho) &=
\xy(0,0)*+{\xy(-1.8,6)*{\psi};(0,1.5)*{};(0,10.5)*{}**\crv{(-6,1.5)&(-6,10.5)};(0,1.5)*{};(0,10.5)*{}**\dir{-};(0,3)*{};(3,3)*{}**\dir{-};(3,-4.5)*{};(12,-4.5)*{}**\dir{-};(12,-4.5)*{};(12,4.5)*{}**\dir{-};(12,4.5)*{};(3,4.5)*{}**\dir{-};(3,4.5)*{};(3,-4.5)*{}**\dir{-};(7.5,0)*{V};(15.3,0)*{\scriptstyle 0};(16.5,-1.5)*{};(16.5,1.5)*{}**\crv{(13.5,-1.5)&(13.5,1.5)};(16.5,-1.5)*{};(16.5,1.5)*{}**\dir{-};(16.5,0)*{};(19.5,0)*{}**\dir{-};(19.5,3)*{};(16.75,3)*{}**\dir{-};{\ar(12,3)*{};(16.75,3)*{}};(19.5,-1.5)*{};(25.5,-1.5)*{}**\dir{-};(25.5,-1.5)*{};(25.5,4.5)*{}**\dir{-};(25.5,4.5)*{};(19.5,4.5)*{}**\dir{-};(19.5,4.5)*{};(19.5,-1.5)*{}**\dir{-};(22.5,1.5)*{U};(28.5,-3)*{};(21.25,-3)*{}**\dir{-};{\ar(12,-3)*{};(21.25,-3)*{}};(25.5,0)*{};(28.5,0)*{}**\dir{-};(28.5,-4.5)*{};(34.5,-4.5)*{}**\dir{-};(34.5,-4.5)*{};(34.5,1.5)*{}**\dir{-};(34.5,1.5)*{};(28.5,1.5)*{}**\dir{-};(28.5,1.5)*{};(28.5,-4.5)*{}**\dir{-};(31.5,-1.5)*{V^\dagger};(34.5,0)*{};(37.5,0)*{}**\dir{-};(37.5,3)*{};(32.5,3)*{}**\dir{-};{\ar(25.5,3)*{};(32.5,3)*{}};(37.5,-1.5)*{};(43.5,-1.5)*{}**\dir{-};(43.5,-1.5)*{};(43.5,4.5)*{}**\dir{-};(43.5,4.5)*{};(37.5,4.5)*{}**\dir{-};(37.5,4.5)*{};(37.5,-1.5)*{}**\dir{-};(40.5,1.5)*{W^\dagger};(43.5,3)*{};(46.5,3)*{}**\dir{-};(46.5,9)*{};(24.25,9)*{}**\dir{-};{\ar(0,9)*{};(24.25,9)*{}};(48.3,6)*{\psi};(46.5,1.5)*{};(46.5,10.5)*{}**\crv{(52.5,1.5)&(52.5,10.5)};(46.5,1.5)*{};(46.5,10.5)*{}**\dir{-};\endxy};\endxy
\\
&= \sum_{ij}
\xy(0,0)*+{\xy(-2.4,8)*{\psi};(0,2)*{};(0,14)*{}**\crv{(-8,2)&(-8,14)};(0,2)*{};(0,14)*{}**\dir{-};(0,4)*{};(4,4)*{}**\dir{-};(5.6,4)*{\scriptstyle i};(4,2)*{};(4,6)*{}**\crv{(8,2)&(8,6)};(4,2)*{};(4,6)*{}**\dir{-};(10.4,4)*{\scriptstyle j};(12,2)*{};(12,6)*{}**\crv{(8,2)&(8,6)};(12,2)*{};(12,6)*{}**\dir{-};(12,4)*{};(16,4)*{}**\dir{-};(16,2)*{};(20,2)*{}**\dir{-};(20,2)*{};(20,6)*{}**\dir{-};(20,6)*{};(16,6)*{}**\dir{-};(16,6)*{};(16,2)*{}**\dir{-};(18,4)*{{\scriptstyle s_i^{\dagger}}};(18.4,0)*{\scriptstyle 0};(20,-2)*{};(20,2)*{}**\crv{(16,-2)&(16,2)};(20,-2)*{};(20,2)*{}**\dir{-};(20,0)*{};(24,0)*{}**\dir{-};(20,4)*{};(24,4)*{}**\dir{-};(24,-2)*{};(32,-2)*{}**\dir{-};(32,-2)*{};(32,6)*{}**\dir{-};(32,6)*{};(24,6)*{}**\dir{-};(24,6)*{};(24,-2)*{}**\dir{-};(28,2)*{U};(32,0)*{};(36,0)*{}**\dir{-};(36,-2)*{};(40,-2)*{}**\dir{-};(40,-2)*{};(40,2)*{}**\dir{-};(40,2)*{};(36,2)*{}**\dir{-};(36,2)*{};(36,-2)*{}**\dir{-};(38,0)*{{\scriptstyle s_i}};(40,0)*{};(44,0)*{}**\dir{-};(45.6,0)*{\scriptstyle j};(44,-2)*{};(44,2)*{}**\crv{(48,-2)&(48,2)};(44,-2)*{};(44,2)*{}**\dir{-};(32,4)*{};(36,4)*{}**\dir{-};(37.6,4)*{\scriptstyle i};(36,2)*{};(36,6)*{}**\crv{(40,2)&(40,6)};(36,2)*{};(36,6)*{}**\dir{-};(44.4,4)*{\scriptstyle i};(46,2)*{};(46,6)*{}**\crv{(42,2)&(42,6)};(46,2)*{};(46,6)*{}**\dir{-};(46,4)*{};(50,4)*{}**\dir{-};(50,12)*{};(26,12)*{}**\dir{-};{\ar(0,12)*{};(26,12)*{}};(52.4,8)*{\psi};(50,2)*{};(50,14)*{}**\crv{(58,2)&(58,14)};(50,2)*{};(50,14)*{}**\dir{-};\endxy};\endxy
\\
&= \sum_i \bra{i} \rho \ket{i}
\xy(0,0)*+{\xy(0,4)*{};(0,8)*{}**\crv{(-3.2,4)&(-3.2,8)};(0,4)*{};(4,4)*{}**\dir{-};(4,2)*{};(8,2)*{}**\dir{-};(8,2)*{};(8,6)*{}**\dir{-};(8,6)*{};(4,6)*{}**\dir{-};(4,6)*{};(4,2)*{}**\dir{-};(6,4)*{{\scriptstyle s_i^{\dagger}}};(6.4,0)*{\scriptstyle 0};(8,-2)*{};(8,2)*{}**\crv{(4,-2)&(4,2)};(8,-2)*{};(8,2)*{}**\dir{-};(8,0)*{};(12,0)*{}**\dir{-};(8,4)*{};(12,4)*{}**\dir{-};(12,-2)*{};(20,-2)*{}**\dir{-};(20,-2)*{};(20,6)*{}**\dir{-};(20,6)*{};(12,6)*{}**\dir{-};(12,6)*{};(12,-2)*{}**\dir{-};(16,2)*{U};(20,0)*{};(24,0)*{}**\dir{-};(24,-2)*{};(28,-2)*{}**\dir{-};(28,-2)*{};(28,2)*{}**\dir{-};(28,2)*{};(24,2)*{}**\dir{-};(24,2)*{};(24,-2)*{}**\dir{-};(26,0)*{{\scriptstyle s_i}};(28,0)*{};(32,0)*{}**\dir{-};(0,8)*{};(15,8)*{}**\dir{-};{\ar(32,8)*{};(15,8)*{}};(32,0)*{};(32,8)*{}**\crv{(38.4,0)&(38.4,8)};(20,4)*{};(24,4)*{}**\dir{-};(25.6,4)*{\scriptstyle i};(24,2)*{};(24,6)*{}**\crv{(28,2)&(28,6)};(24,2)*{};(24,6)*{}**\dir{-};\endxy};\endxy
\\
&= \sum_i \bra{i} \rho \ket{i}
\xy(0,0)*+{\xy(0,4)*{};(0,8)*{}**\crv{(-3.2,4)&(-3.2,8)};(0,4)*{};(4,4)*{}**\dir{-};(4,2)*{};(8,2)*{}**\dir{-};(8,2)*{};(8,6)*{}**\dir{-};(8,6)*{};(4,6)*{}**\dir{-};(4,6)*{};(4,2)*{}**\dir{-};(6,4)*{{\scriptstyle \rho_i}};(6.4,0)*{\scriptstyle 0};(8,-2)*{};(8,2)*{}**\crv{(4,-2)&(4,2)};(8,-2)*{};(8,2)*{}**\dir{-};(8,0)*{};(12,0)*{}**\dir{-};(8,4)*{};(12,4)*{}**\dir{-};(12,-2)*{};(20,-2)*{}**\dir{-};(20,-2)*{};(20,6)*{}**\dir{-};(20,6)*{};(12,6)*{}**\dir{-};(12,6)*{};(12,-2)*{}**\dir{-};(16,2)*{U};(20,0)*{};(24,0)*{}**\dir{-};(0,8)*{};(11,8)*{}**\dir{-};{\ar(24,8)*{};(11,8)*{}};(24,0)*{};(24,8)*{}**\crv{(30.4,0)&(30.4,8)};(20,4)*{};(24,4)*{}**\dir{-};(25.6,4)*{\scriptstyle i};(24,2)*{};(24,6)*{}**\crv{(28,2)&(28,6)};(24,2)*{};(24,6)*{}**\dir{-};\endxy};\endxy
= \tr(UX^{\dagger}),\\
\end{split}
\end{equation}
with $\ket{\psi}$ being a purification of $\rho$, and
\begin{equation}
X = \sum_{i} \bra{i} \rho \ket{i}\, \rho_i \otimes \ket i \bra 0.
\end{equation}
In the above, we used the notation
\begin{equation}
\begin{split}
\xy(0,0)*+{\xy(0,0)*{};(0,8)*{}**\crv{(-6.4,0)&(-6.4,8)};(0,0)*{};(4,0)*{}**\dir{-};(4,-2)*{};(8,-2)*{}**\dir{-};(8,-2)*{};(8,2)*{}**\dir{-};(8,2)*{};(4,2)*{}**\dir{-};(4,2)*{};(4,-2)*{}**\dir{-};(6,0)*{A};(8,0)*{};(12,0)*{}**\dir{-};(0,8)*{};(5,8)*{}**\dir{-};{\ar(12,8)*{};(5,8)*{}};(12,0)*{};(12,8)*{}**\crv{(18.4,0)&(18.4,8)};\endxy};\endxy
&= \sum_i
\xy(0,0)*+{\xy(0,0)*{};(0,8)*{}**\crv{(-6.4,0)&(-6.4,8)};(0,8)*{};(4,8)*{}**\dir{-};(5.6,8)*{\scriptstyle i};(4,6)*{};(4,10)*{}**\crv{(8,6)&(8,10)};(4,6)*{};(4,10)*{}**\dir{-};(10.4,8)*{\scriptstyle i};(12,6)*{};(12,10)*{}**\crv{(8,6)&(8,10)};(12,6)*{};(12,10)*{}**\dir{-};(0,0)*{};(4,0)*{}**\dir{-};(4,-2)*{};(8,-2)*{}**\dir{-};(8,-2)*{};(8,2)*{}**\dir{-};(8,2)*{};(4,2)*{}**\dir{-};(4,2)*{};(4,-2)*{}**\dir{-};(6,0)*{A};(16,0)*{};(13,0)*{}**\dir{-};{\ar(8,0)*{};(13,0)*{}};(12,8)*{};(16,8)*{}**\dir{-};(16,0)*{};(16,8)*{}**\crv{(22.4,0)&(22.4,8)};\endxy};\endxy\\
&= \sum_i
\xy(0,0)*+{\xy(-1.6,0)*{\scriptstyle i};(0,-2)*{};(0,2)*{}**\crv{(-4,-2)&(-4,2)};(0,-2)*{};(0,2)*{}**\dir{-};(0,0)*{};(4,0)*{};**\dir{-};(4,-2)*{};(8,-2)*{}**\dir{-};(8,-2)*{};(8,2)*{}**\dir{-};(8,2)*{};(4,2)*{}**\dir{-};(4,2)*{};(4,-2)*{}**\dir{-};(6,0)*{A};(8,0)*{};(12,0)*{};**\dir{-};(13.6,0)*{\scriptstyle i};(12,-2)*{};(12,2)*{}**\crv{(16,-2)&(16,2)};(12,-2)*{};(12,2)*{}**\dir{-};\endxy};\endxy
= \tr A.
\end{split}
\end{equation}
We also used that fact that
\begin{equation}
\begin{split}
\xy(0,0)*+{\xy(-2,2)*{\psi};(0,-2)*{};(0,6)*{}**\crv{(-6,-2)&(-6,6)};(0,-2)*{};(0,6)*{}**\dir{-};(0,0)*{};(4,0)*{}**\dir{-};(5.6,0)*{\scriptstyle i};(4,-2)*{};(4,2)*{}**\crv{(8,-2)&(8,2)};(4,-2)*{};(4,2)*{}**\dir{-};(10.4,0)*{\scriptstyle i};(12,-2)*{};(12,2)*{}**\crv{(8,-2)&(8,2)};(12,-2)*{};(12,2)*{}**\dir{-};(12,0)*{};(16,0)*{}**\dir{-};(16,4)*{};(9,4)*{}**\dir{-};{\ar(0,4)*{};(9,4)*{}};(18,2)*{\psi};(16,-2)*{};(16,6)*{}**\crv{(22,-2)&(22,6)};(16,-2)*{};(16,6)*{}**\dir{-};\endxy};\endxy
&=
\xy(0,0)*+{\xy(-1.6,0)*{\scriptstyle i};(0,-2)*{};(0,2)*{}**\crv{(-4,-2)&(-4,2)};(0,-2)*{};(0,2)*{}**\dir{-};(0,4)*{};(0,8)*{}**\crv{(-3.2,4)&(-3.2,8)};(0,0)*{};(4,0)*{}**\dir{-};(0,4)*{};(4,4)*{}**\dir{-};(6,2)*{\psi};(4,-2)*{};(4,6)*{}**\crv{(10,-2)&(10,6)};(4,-2)*{};(4,6)*{}**\dir{-};(14,2)*{\psi};(16,-2)*{};(16,6)*{}**\crv{(10,-2)&(10,6)};(16,-2)*{};(16,6)*{}**\dir{-};(16,4)*{};(20,4)*{}**\dir{-};(0,8)*{};(9,8)*{}**\dir{-};{\ar(20,8)*{};(9,8)*{}};(20,4)*{};(20,8)*{}**\crv{(23.2,4)&(23.2,8)};(16,0)*{};(20,0)*{}**\dir{-};(21.6,0)*{\scriptstyle i};(20,-2)*{};(20,2)*{}**\crv{(24,-2)&(24,2)};(20,-2)*{};(20,2)*{}**\dir{-};\endxy};\endxy\\
&=
\xy(0,0)*+{\xy(-1.6,0)*{\scriptstyle i};(0,-2)*{};(0,2)*{}**\crv{(-4,-2)&(-4,2)};(0,-2)*{};(0,2)*{}**\dir{-};(0,0)*{};(4,0)*{}**\dir{-};(4,-2)*{};(8,-2)*{}**\dir{-};(8,-2)*{};(8,2)*{}**\dir{-};(8,2)*{};(4,2)*{}**\dir{-};(4,2)*{};(4,-2)*{}**\dir{-};(6,0)*{\rho};(8,0)*{};(12,0)*{}**\dir{-};(13.6,0)*{\scriptstyle i};(12,-2)*{};(12,2)*{}**\crv{(16,-2)&(16,2)};(12,-2)*{};(12,2)*{}**\dir{-};\endxy};\endxy
= \bra i \rho \ket i.
\end{split}
\end{equation}
Therefore,
\begin{equation}
X^\dagger X = \sum_{i}  \rho_i^2 \bra{i} \rho \ket{i}^2 \otimes \proj 0
\end{equation}
and
\begin{equation}
F(\cchan N, \cchan N \cchan M) = \min_\rho \tr\sqrt{X^\dagger X}. 
\end{equation}
Hence, the quantity
\begin{equation}
\Delta := 1 - \min_p \tr \sqrt{\sum\nolimits_i p_i^2 \rho_i^2}
\end{equation}
satisfies
\begin{equation}
\frac{1}{4} \Delta \le \min_A \max_i \left[{ 1 - \sqrt{\tr( \rho_i A_i )}}\right] \le \Delta,
\end{equation}
i.e.,
\begin{equation}
\frac{1}{2}\Delta - \frac{1}{16} \Delta^2 \le \min_A \max_i \left[{ 1 - \tr( \rho_i A_i )}\right] \le 2\Delta - \Delta^2.
\end{equation}
This provides a simple estimate to the optimal achievable solution to the minimax state discrimination problem.

We note that the same upper bound, and a better lower bound exactly equal to $\Delta$, can also be derived \cite{tyson09x1} by applying the minimax theorem to previously obtained state discrimination bounds~\cite{tyson09,ogawa99,tyson09err}.





\section{Comparison with other results}
\label{comparison}

Results similar to ours exist in the special case $\chan M=\id$~\cite{barnum00, tyson09x1}, or in the particular case of state discrimination~\cite{tyson09}. In these works, bounds are derived for the entanglement fidelity for a fixed state, but a direct application of the minimax theorem yields bounds for the worst-case entanglement fidelity~\cite{tyson09x1}. However, it is not known whether there exists an efficient procedure for constructing near-optimal recovery channels compatible with the worst-case bounds obtained in this way.

Let us show that our method also works for the fixed-state entanglement fidelity, at least for the main theorem. Then we will see that it yields almost the same bounds as in Ref.~\cite{tyson09x1} in the case $\chan M=\id$, albeit weaker.

It is easy to show that both Theorems \ref{main} and \ref{main2} still hold if the worst-case fidelity $F$ is replaced by the fidelity $F_\rho$ for a fixed input state $\rho$~\cite{beny10}. The proofs are much simpler as the minimum disappears from Eq.~\eqref{minimax1}. Hence, it suffices to see that $\max_{\|A\| \le 1, \|A'\|\le 1} g(A,A')$ is equal to both sides of the equation
\begin{equation}
\max_{|\chan R| \le d} F_\rho(\chan R\chan N,\chan M) = \max_{|\chan R'| \le d}F_\rho(\cchan N,\chan R' \cchan M)
\end{equation}
when $\cchan N$ and $\cchan M$ are complementary, respectively, to $\chan N$ and $\chan M$. From this it follows also that if $\gcchan N \sim \cchan N$ and $\gcchan M \sim \cchan M$, then
\begin{equation}
\max_{\chan R} F_\rho(\chan R\chan N,\chan M) = \max_{\chan R'}F_\rho(\gcchan N,\chan R' \gcchan M).
\end{equation}
However, since $F_\rho$ does not have the special property expressed in Eq.~\eqref{rmonot}, we cannot use the same technique to get a simple approximation of $\max_{\chan R'}F_\rho(\gcchan N,\chan R' \gcchan M)$ in the case $\gcchan M^2 = \gcchan M$. However, we can obtain an inequality similar to Eq.~\eqref{bounds} for the important case
\begin{equation}
\gcchan M(\sigma) = \rho \tr(\sigma),
\end{equation}
where $\rho$ is the same state as the one used to evaluate the fidelity. This channel $\gcchan M$ is generalized complementary to $\chan M = \id$. Hence, this corresponds to the approximate quantum error correction problem.

Concretely, suppose that $\chan R'_0$ is such that
\begin{equation}
d_\rho(\gcchan N,\chan R'_0(\rho) \tr) = \min_{\chan R'}d_\rho(\gcchan N,\chan R'(\rho) \tr)
\equiv \varepsilon_0.
\end{equation}
Then, using the triangle inequality, we have
\begin{equation}
d_\rho(\gcchan N,\gcchan N(\rho) \tr) \le \varepsilon_0 + d_\rho(\chan R'_0(\rho) \tr,\gcchan N(\rho) \tr).
\end{equation}
The second term is calculated from
\begin{equation}
\begin{split}
F_\rho(\chan R'_0(\rho) \tr,\gcchan N(\rho) \tr) &= f(\chan R'_0(\rho),\gcchan N(\rho))\\
&\le F_\rho (\gcchan N, \chan R'_0(\rho) \tr) = 1-\varepsilon_0^2.
\end{split}
\end{equation}
Hence, $d_\rho(\gcchan N,\gcchan N(\rho) \tr) \le 2 \varepsilon_0$, from which we obtain the estimate
\begin{equation}
\frac 1 2 d_\rho(\gcchan N,\gcchan N(\rho) \tr) \le \min_{\chan R} d_\rho(\chan R\chan N,\id) \le d_\rho(\gcchan N,\gcchan N(\rho) \tr).
\end{equation}

This is a weaker form of the bounds derived by Tyson (Eq.~(153) of Ref.~\cite{tyson09x1}) using a different method. The upper bound can be seen to be equivalent to the corresponding bound in Ref.~\cite{tyson09x1}, but the lower bound is weaker, which may be significant in the regime where the optimal error is large.

Indeed, if we write the estimate explicitly in terms of some Kraus operators $E_i$ of $\chan N$, we have
\begin{equation}
F_\rho(\gcchan N,\gcchan N(\rho) \tr) = \tr \sqrt{ \sum_{ij} E_i \rho^2 E_j^\dagger \tr(\rho E_i^\dagger E_j) }.
\end{equation}
This is precisely the quantity that Tyson denotes by $\Lambda$ (Eq.~(154) of Ref.~\cite{tyson09x1}) (note that the quantity that Tyson calls `fidelity' is the square of our fidelity). Tyson writes the channel as $\chan N(\rho) = \sum_i p_i F_i \rho F_i^\dagger$, where $\tr(\rho F_i^\dagger F_j) = \delta_{ij}$ and $p_i>0$, which is always possible. Hence, we obtain his formula for $\Lambda$ using $E_i = \sqrt p_i F_i$.

The corresponding near-optimal channel that we obtain in this way is the same as the one introduced by Tyson, which is $\chan R_g$ defined in Eq.~\eqref{goodchannel} but with $\rho_0 = \sigma = \rho$, i.e.,
\begin{equation}
\chan R_g(\tau) = \Phi^\dagger[\Phi(\one)^{- \frac 1 2}\,\tau\,\Phi(\one)^{- \frac 1 2} ] + \chan T (\tau),
\end{equation}
where
\begin{equation}
\Phi^\dagger(\tau) := \sum_{ij} \tr( \rho E_i^\dagger E_j ) \rho E_j^\dagger \tau E_i \rho,
\end{equation}
and $\chan T$ is any CP map that makes $\chan R_g$ trace-preserving.
For instance, $\chan T$ can be chosen as
\begin{equation}
\chan T(\tau) = \tr( \tau P ) \sigma,
\end{equation}
where $P$ projects on the kernel of $\Phi(\one)$, and $\sigma$ can be any state.

As explained in Ref.~\cite{tyson09x1}, this channel is not the same as the one used in Ref.~\cite{barnum00}, which yields similar bounds and was introduced by Petz~\cite{petz88,hiai10} who showed that it yields exact inversion on two given states. The latter is built in the same way, but from the CP map
\begin{equation}
\Phi^\dagger(\tau) = \sqrt \rho \chanh N(\tau) \sqrt \rho.
\end{equation}
The performance of this channel with $\rho$ maximally mixed (known as the ``transpose channel'') was also studied for approximate QEC in terms of the worst-case fidelity in Ref.~\cite{ng09}.

\section{Acknowledgements} CB would like to thank Milan Mosonyi and Jon Tyson for related discussions. OO was supported by the Spanish MICINN (Consolider-Ingenio QOIT), the Foundational Questions Institute (FQXi), and the Interuniversity Attraction Poles program of the Belgian Science Policy Office, under grant IAP P6-10 $\ll$photonics@be$\gg$. This work was supported in part by the cluster of excellence EXC 201 ``Quantum Engineering and Space-Time Research''. CB also acknowledges the support by the EU projects CORNER and COQUIT. The Centre for Quantum Technologies is funded by the Singapore Ministry of Education and the National Research Foundation as part of the Research Centres of Excellence program.

\bibliography{approx_qec3}

\newpage

\appendix


\section{Proof of the relation $\re\tr(V^\dagger A |X_\rho|)\leq \tr|X_\rho|$}
\label{firstapp}

Let $\{\ket{i}\}$ be an eigenbasis of $|X_\rho|$, $|X_\rho|=\sum_jx_j|j\rangle\langle j|$, $x_j\geq 0$. Since $V$ is unitary and $\|A\|\leq 1$, we have $\|V^\dagger A\|\leq \|V^\dagger\|\|A\|\leq 1$. This means that for all $j$,
\begin{gather}
\| V^\dagger A \ket{j} \|^2 = \bra{j} A^\dagger V V^\dagger A \ket{j} = \sum_i \bra{j} A^\dagger V \proj i V^\dagger A \ket{j}\le 1.
\end{gather}
Therefore,
\begin{equation}
|\bra j V^\dagger A \ket{j}|\leq 1, \quad \forall j.
\end{equation}
We thus have
\begin{gather}
\re\tr(V^\dagger A |X_\rho|)=\re\sum_j\bra j V^\dagger A \ket{j}x_j\leq \sum_j|\bra j V^\dagger A \ket{j}|x_j\notag\\
\leq \sum_jx_j\equiv \tr|X_\rho|.\notag
\end{gather}

\section{The saddle point in the case of a unique, full-rank $\rho_0$}
\label{obvious}

Our procedure for constructing a near-optimal recovery channel requires finding a saddle point $(\rho_0,A_0)$ of $\re g_{\rho}(A,U')$, where $U'$ yields $\cchan N$ through $\cchan N(\rho) = \tr_2(U'(\rho \otimes \proj{0}) (U')^\dagger)$. In the case of $\chan M(\rho)=\rho\otimes\sigma$, we saw that this is equivalent to finding $(\rho_0,A_0)$ that achieves the optimization $\min_\rho \max_{\|A\|\le 1} \re \tr(A X_\rho^\dagger) =  \min_\rho \tr |X_\rho|$. One way to approach the problem in this case could be to first search for $\rho_0$ that achieves $\min_{\rho}\tr|X_{\rho}|$, which is a convex optimization task. If we find $\rho_0$ that is unique, then we know that it must be the one at the saddle point. Now imagine that this $\rho_0$ is also of full rank. We will prove that in such a case the staddle-point $A_0$ can be taken to be $A_0=U_0$, where $U_0$ is any unitary that comes from the polar decomposition of $X_{\rho_0}$, $X_{\rho_0}=U_0|X_{\rho_0}|$.

Clearly, the unitary $U_0$ achieves the maximum in $\max_{\|A\|\le 1} \re \tr(A X_{\rho_0}^\dagger)$, because $\re\tr(U_0X_{\rho_0}^\dagger)=\re\tr(U_0|X_{\rho_0}|U_0^{\dagger})=\tr|X_{\rho_0}|$, but we also need that $\rho_0$ achieves the minimum in $\min_\rho\re \tr(U_0 X_\rho^\dagger)$. If $U_0$ is the unique maximizer of $\max_{\|A\|\le 1} \re \tr(A X_{\rho_0}^\dagger)$, then we know that it must be a saddle point. The problem is that in general $A_0$ need not be unique. However, we will see that if $\rho_0$ is of full rank, $A_0$ is unique up to a freedom that is irrelevant for the value of $\re\tr(AX_{\rho}^\dagger)$ whose saddle point we are looking for. Hence, any operator $A_0$ that maximizes $\re \tr(A X_{\rho_0}^\dagger)$ would yield a saddle point.

To show this, let us characterize the operators $A$, $\|A\| \le 1$, that satisfy
\begin{gather}
\re \tr(A X_{\rho_0}^\dagger) = \tr|X_{\rho_0}|. \label{AXX} 
\end{gather}
Eq.~\eqref{AXX} can be equivalently written as
\begin{gather}
\re \tr (\check{A} |X_{\rho_0}|) = \tr|X_{\rho_0}|,\label{AXX2}
\end{gather}
where (using the cyclic invariance of the trace)
\begin{equation}
\check{A} :=  U_0^{\dagger}A.
\end{equation}
Note that $\|\check{A}\| \le \|U_0^{\dagger}\|\|A\|\le 1$.

Let $\ket{i}$, $i=1,\dots,d$, be eigenvectors of $|X_{\rho_0}|$ ordered such that the first $n\leq d$ of them are all those with nonzero eigenvalues, $|X_{\rho_0}| = \sum_{i\le n} x_i \proj{i}$, $x_i > 0$. Then Eq.~\eqref{AXX2} is equivalent to
\begin{equation}
\re \sum_{i \le n} \bra{i} \check{A} \ket{i} x_i = \sum_{i\le n} x_i,
\end{equation}
or
\begin{equation}
\sum_{i \le n} (1 - \re \bra{i} \check{A} \ket{i}) x_i = 0.
\end{equation}
Since $\|\check{A}\| \le 1$, we have $|\bra{i} \check{A} \ket{i}| \le 1$ for all $i \le n$. Therefore, the above is satisfied if and only if
\begin{equation}
\bra{i} \check{A} \ket{i} = 1, \quad \forall i \le n.
\end{equation}
From $\|\check{A}\| \le 1$ we also have that for all $j \le n$,
\begin{equation}
\begin{split}
\| \check{A} \ket{j} \|^2 &= \bra{j} \check{A}^\dagger \check{A} \ket{j} = \sum_i \bra{j} \check{A}^\dagger \proj i \check{A} \ket{j}\\
&= 1 + \sum_{i \neq j} \bra{j} \check{A}^\dagger \proj i \check{A} \ket{j} \le 1,\\
\end{split}
\end{equation}
which is only possible if
\begin{equation}
\bra i \check{A} \ket j = 0, \quad \forall \; i\neq j \, ,\; j\le n.
\end{equation}
Since $\|\check{A}^\dagger \| = \|\check{A}\| \le 1$, we obtain via the same argument
\begin{equation}
\bra i \check{A} \ket j = 0, \quad \forall \; i\neq j \, ,\; i\le n.
\end{equation}
This implies that $\check{A}$ is block diagonal in the basis $\{\ket{i}\}$, with the upper block (corresponding to the first $n$ basis vectors) equal to $\one$. This is necessary and sufficient for our condition to hold. Let us write
\begin{equation}
\check{A} = \begin{pmatrix} \one & 0 \\ 0 & {B} \end{pmatrix}
\end{equation}
for some matrix ${B}$ and $\one = \sum_{i \le n} \proj{i}$.

Coming back to $A$ itself, and labeling the blocks of $U_0$ by $U_{\mu \nu}$, we have that
\begin{equation}
A = U_0\check{A} = \begin{pmatrix} U_{11} & U_{12} \\ U_{21} & U_{22} \end{pmatrix} \begin{pmatrix} \one & 0 \\ 0 & {B} \end{pmatrix}
= \begin{pmatrix} U_{11} & U_{12}{B} \\  U_{21} & U_{22}{B} \end{pmatrix}.
\end{equation}
Since the left block column of $A$ is unique, we know that it is equal to that of $A_0$ which sits at the saddle point.

Now, suppose that we replace $X_{\rho_0}$ by $X_{\rho}$ which is such that the support of $|X_{\rho}|$ is within the support of $|X_{\rho_0}|$. Then let us show that only the left block column of $A$ as defined above would matter for the calculation of the pseudo-fidelity
\begin{equation}
\re \tr(AX_{\rho}^\dagger) = \re \tr(V^\dagger A|X_\rho|),
\end{equation}
where $X_{\rho}=V|X_\rho|$. Indeed, since the only nonzero components of $|X_\rho|$ are in the upper left block, for the trace in the last expression only the upper left block of $V^\dagger A$ would matter, and this block is
\begin{equation}
(V^\dagger A)_{11} = V_{11}^\dagger A_{11}  + V_{21}^\dagger A_{21}  = V_{11}^\dagger U_{11}  + V_{21}^\dagger U_{21},
\end{equation}
where $V_{\mu\nu}$ are the corresponding blocks of $V$. We see that $(V^\dagger A)_{11}$ is independent of any freedom we may have in choosing $A$ and therefore behaves just like for the saddle-point $A_0$.
\qed

For $\chan M(\rho)=\rho\otimes\sigma$, we can write $|X_{\rho_0}|^2$ in the form
\begin{gather}
|X_{\rho_0}|^2 = \sum_i E_i \rho_0^2 E_i^{\dagger} \lambda_i,
\end{gather}
where $E_i$ are suitable Kraus operators of $\chan N$, and $\lambda_i=\tr(E_i\sigma E_i^{\dagger})$ [$\cchan M(\rho)=\sigma \tr(\rho)$]. It is easy to see that a
state $|\psi\rangle$ is in the kernel of $|X_{\rho_0}|^2$ (and therefore of $|X_{\rho_0}|$) if and only if, for all $i$,
$\lambda_i = 0$  or $\rho_0 E_i^{\dagger} |\psi\rangle = 0$. If $\rho_0$ is full-rank, the last condition reads $E_i^{\dagger} |\psi\rangle = 0$. This means that for any $\rho$, the kernel of $|X_\rho|$ contains the kernel of $|X_{\rho_0}|$, or equivalently, the support of $|X_{\rho_0}|$ contains the support of $|X_\rho|$ for all $\rho$. By the above argument, $A_0=U_0$ would be the unique maximizer of $\re\tr(AX^\dagger_{\rho_0})$ up to the freedom in the way $\check{A}=U_0^\dagger A$ acts on the kernel of $|X_\rho|$, which has no relevance for the value of $\re\tr(\check{A}|X_\rho|)=\re\tr(AX^\dagger_{\rho_0})$. Hence, $(U_0,\rho_0)$ is a saddle point of $\re\tr(AX^\dagger_{\rho})$.


\section{Inadequacy of the above procedure when $\rho_0$ is not of full rank}
\label{counter}


Unfortunately, the above argument cannot be used to simplify the procedure in the general case, since in principle $\rho_0$ need not be unique (take, for example, the extreme case where $\chan N$ is correctable), and even if it is unique, it need not be of full rank. Let us illustrate the latter case by an example.

Let $\cchan N$ be a channel with a 2-dimensional input (we will denote the input system by $A$) and a 2-dimensional output (denoted by $E$) with basis vectors $\{|0\rangle^A,|1\rangle^A\}$ and $\{|0\rangle^E,|1\rangle^E\}$, respectively, that acts as follows: $\cchan N (\rho^A)=(1-s)\rho^E+ s|0\rangle\langle 0|^E\tr(\rho^A)$. Take $\cchan M (\rho^A)=|0\rangle\langle 0|^E\tr(\rho^A)$. Note that $\cchan N=(1-s)\mathbf{1}+s\cchan M$ and $\cchan N \cchan M = \cchan M$. Let $|\psi_{\rho}\rangle^{AR}$ be a purification of $\rho^A$. From the concavity of the \textit{square} of the fidelity \cite{uhlmann76}, we have
\begin{gather}
\min_\rho F^2_{\rho}(\cchan N,\cchan N\cchan M)\notag\\=\min_{\rho}F^2((1-s)|\psi\rangle\langle\psi|^{ER}+s\cchan M\otimes\id^R(|\psi\rangle\langle\psi|^{AR}),\notag\\\cchan M\otimes\id^R(|\psi\rangle\langle\psi|^{AR}))\notag\\
\geq \min_{\rho}[(1-s)F^2(|\psi\rangle\langle\psi|^{ER},M\otimes\id^R(|\psi\rangle\langle\psi|^{AR}))+\notag\\
s F^2(M\otimes\id^R|\psi\rangle\langle\psi|^{AR},M\otimes\id^R(|\psi\rangle\langle\psi|^{AR}))]\geq s.
\end{gather}
We will show that the lower bound $s$ is actually achievable for $\rho^A_0=|1\rangle\langle 1|^A$. Indeed, in this case we can take $|\psi\rangle^{AR}=|1\rangle^A|1\rangle^R$. We then have $\chan M\otimes\id^R(|\psi\rangle\langle\psi|^{AR})=|0\rangle\langle 0|^E\otimes |1\rangle\langle 1|^R$. We obtain
\begin{gather}
F^2_{|1\rangle\langle 1|}(\cchan N,\cchan N\cchan M)\notag\\
=\langle 0|^E\langle 1|^R((1-s)|1\rangle\langle 1|^E\otimes |1\rangle\langle 1|^R+\notag\\
s|0\rangle\langle 0|^E\otimes |1\rangle\langle 1|^R)|0\rangle^E|1\rangle^R=s.
\end{gather}
Moreover, it is easy to see that $\rho^A_0=|1\rangle\langle 1|^A$ is the unique state that achieves the minimum value.

The state $\rho^A_0=|1\rangle\langle 1|^A$ does not have full support. To show that this does not allow us to obtain a saddle-point $A_0$ by simply taking any maximizer of $\max_Ag_{\rho_0}(A,U')$, let us look at the support of $|X_{\rho_0}|$ as a function of $\rho$. Since $\cchan N$ has three Kraus operators, $\hat{E}_0=\sqrt{1-s}\mathbf{1}$, $\hat{E}_1=\sqrt{s}|0\rangle\langle 0|$, $\hat{E}_2=\sqrt{s}|0\rangle\langle 1|$, we will take system $B$ in the circuit diagram to be 3-dimensional, with basis $\{|0^B\rangle, |1^B\rangle, |2^B\rangle\}$.
The dilation of $\cchan N$ (or $\chan N$) is
\begin{gather}
|0\rangle^A \rightarrow \sqrt{1-s}|0\rangle^E|0\rangle^B+\sqrt{s}|0\rangle^E|1\rangle^B,\\
|1\rangle^A \rightarrow \sqrt{1-s}|1\rangle^E|0\rangle^B+\sqrt{s}|0\rangle^E|2\rangle^B.\label{isomexample}
\end{gather}
From this, we obtain the Kraus operators of $\chan N$,
\begin{gather}
{E}_0=\sqrt{1-s}|0\rangle\langle 0|+\sqrt{s}|1\rangle\langle 0|+\sqrt{s}|2\rangle\langle 1|,\\
{E}_1=\sqrt{1-s}|0\rangle\langle 1|.
\end{gather}
Using the expression for ${X^{\dagger}_{\rho}X_{\rho}}$ in terms of an arbitrary choice of Kraus operators,
\begin{gather}
{X^{\dagger}_{\rho}X_{\rho}}=\sum_{ij}E_i\rho^2E_j^{\dagger}\tr(E_j\sigma E_i^{\dagger}),
\end{gather}
we obtain (in our case $\sigma=|0\rangle\langle 0|$)
\begin{gather}
{X^{\dagger}_{\rho}X_{\rho}}=E_0\rho^2E_0^{\dagger}.
\end{gather}
For $\rho=\rho_0=|1\rangle\langle 1|$, we have
\begin{gather}
|X_{\rho_0}|^2=s|2\rangle\langle 2|.
\end{gather}
However, if we take, for example, $\rho=|0\rangle\langle 0|$, we obtain
\begin{gather}
{X^{\dagger}_{|0\rangle\langle 0|}X_{|0\rangle\langle 0|}}=|\phi\rangle\langle \phi|,
\end{gather}
where
\begin{gather}
|\phi\rangle=\sqrt{1-s}|0\rangle+\sqrt{s}|1\rangle.
\end{gather}
We therefore see that the support of $|X_{\rho_0}|$ in this case does not contain the support of $|X_\rho|$ for all $\rho$.

Let us show that due to this fact, not every unitary one obtains from the polar decomposition of $X^{\dagger}_{\rho_0}$ is a saddle point. Denote the left and right ancilla systems (in state $|0\rangle$) displayed in the main circuit diagram by $C$ (of dimension 2) and $C'$ (of dimension 3), respectively, and denote the system corresponding to the middle wire by $D$ (the latter has dimension 3). Now, consider the case $\rho^A=|1\rangle\langle 1|^A$ (the state $|\psi_{\rho}\rangle^{AR}$ can be taken to be $|1\rangle^A|1\rangle^R$). The action of the isometry $V_{\chan N}$ is given by \eqref{isomexample}, and $U'$ in this case realizes this isometry on the input $E'$ when the ancilla $C'$ is in state $|0\rangle^{C'}$. The input (from the right) at $E'$ is in the state $|0\rangle^{E'}$ since by definition this is the output of $\cchan M$. Using that, one easily obtains that the overlap reduces to
\begin{gather}
g_{|1\rangle\langle 1|}(U,U')=\sqrt{s}\langle 1|^{B'}\langle\phi|^{D}U|2\rangle^B|0\rangle^C.
\end{gather}
The real part of $g_{|1\rangle\langle 1|}(U,U')$ is maximized when
\begin{gather}
U|2\rangle^B|0\rangle^C=|1\rangle^{B'}|\phi\rangle^{D}.\label{conduni}
\end{gather}
However, we have a freedom of choosing how $U$ acts on $\textrm{Span}\{ |0\rangle^B|0\rangle^C, |1\rangle^B|0\rangle^C \}$. To see that not every unitary satisfying \eqref{conduni} yields a saddle point, consider the case of $\rho^A=|0\rangle\langle 0|^A$. In this case, the overlap reduces to
\begin{gather}
g_{|0\rangle\langle 0|}(U,U')=\langle 0|^{B'}\langle\phi|^{D}U|\phi\rangle^B|0\rangle^C.
\end{gather}
Since the action of $U$ on $|\phi\rangle^B|0\rangle^C$ is completely undetermined by condition \eqref{conduni}, we could choose $U_0$ such that it satisfies both $\eqref{conduni}$ and, e.g.,
\begin{gather}
U_0|\phi\rangle^B|0\rangle^C=-|0\rangle^{B'}|\phi\rangle^{D}.
\end{gather}
Then we obtain that $\re g_{|0\rangle\langle 0|}(U_0,U')=-1<\re g_{|1\rangle\langle 1|}(U_0,U')=\sqrt{s}$, i.e., $(U_0,|1\rangle\langle 1|)$ is not a saddle point of $\re g_{\rho}(A,U')$ since $\rho_0=|1\rangle\langle 1|$ does not minimize $\re g_{\rho}(U_0,U')$.

\end{document}